\documentclass[reprint,amsfonts,amssymb,amsmath,showkeys,nofootinbib,pra,superscriptaddress,twocolumn,longbibliography,floatfix]{revtex4-2}
\usepackage[utf8]{inputenc}
\usepackage{graphics}
\usepackage[english]{babel}

\makeatletter 
\renewcommand\onecolumngrid{
\do@columngrid{one}{\@ne}
\def\set@footnotewidth{\onecolumngrid}
\def\footnoterule{\kern-6pt\hrule width 1.5in\kern6pt}
}
\renewcommand\twocolumngrid{
        \def\footnoterule{
        \dimen@\skip\footins\divide\dimen@\thr@@
        \kern-\dimen@\hrule width.5in\kern\dimen@}
        \do@columngrid{mlt}{\tw@}
}
\makeatother

\usepackage{palatino}
\usepackage{mathpazo}
\usepackage[T1]{fontenc}
\usepackage{graphicx} 
\usepackage{mathtools}
\usepackage{adjustbox}
\usepackage{placeins}
\usepackage{csquotes}
\usepackage{physics}
\usepackage{amsthm}
\usepackage{float}
\usepackage[tight,footnotesize]{subfigure}
\usepackage[dvipsnames]{xcolor}
\usepackage{tikz}
\usetikzlibrary{quantikz2}
\colorlet{mdtRed}{red!50!black}
\usepackage[colorlinks=true,citecolor=OliveGreen,linkcolor=magenta,urlcolor=mdtRed]{hyperref}
\usepackage[left=16mm,right=16mm,top=35mm,columnsep=15pt]{geometry}
\usepackage[toc,page]{appendix}

\newtheorem{theorem}{Theorem}
\newtheorem{lemma}{Lemma}
\newtheorem{observation}{Observation}
\newtheorem{corollary}{Corollary}

\newcommand{\Var}{{\rm Var}}
\newcommand{\mse}[1]{\text{MSE}[#1]}
\renewcommand{\vec}[1]{\boldsymbol{#1}}
\renewcommand{\epsilon}{\varepsilon}
\newcommand{\Ntot}{N_{\text{tot}}}
\newcommand{\thv}{\vec{\theta}}
\newcommand{\alv}{\vec{\alpha}}
\newcommand{\nv}{\vec{n}}
\newcommand{\sv}{\vec{s}}
\newcommand{\tv}{\vec{t}}
\newcommand{\xv}{\vec{x}}
\newcommand{\kv}{\vec{k}}
\newcommand{\av}{\vec{a}}
\newcommand{\be}{\begin{equation}}
\newcommand{\ee}{\end{equation}}
\newcommand{\fexpval}[3]{\bra{#1}#2\ket{#3}}

\newcommand{\EC}{\mathcal{E}}
\newcommand{\FC}{\mathcal{F}}
\newcommand{\WC}{\mathcal{W}}
\newcommand{\VC}{\mathcal{V}}
\newcommand{\UC}{\mathcal{U}}
\newcommand{\HC}{\mathcal{H}}

\newcommand{\OC}{\mathcal{O}}
\newcommand{\PC}{\mathcal{P}}
\newcommand{\Cbb}{\mathbb{C}}
\newcommand{\Ebb}{\mathbb{E}}
\newcommand{\Ibb}{\mathbb{I}}
\newcommand{\Rbb}{\mathbb{R}}
\newcommand{\Nbb}{\mathbb{N}}
\newcommand{\Ubb}{\mathbb{U}}

\newcommand{\ad}{^{\dagger}}

\begin{document}

\title{Exact gradients for linear optics with single photons}
\author{Giorgio Facelli}
\affiliation{ORCA Computing, London, UK}
\affiliation{Institute of Physics, \'{E}cole Polytechnique F\'{e}d\'{e}rale de Lausanne (EPFL), Lausanne, Switzerland}

\author{David D. Roberts}
\affiliation{ORCA Computing, London, UK}

\author{Hugo Wallner}
\affiliation{ORCA Computing, London, UK}

\author{Alexander Makarovskiy}
\affiliation{ORCA Computing, London, UK}

\author{Zo\"{e} Holmes}
\affiliation{Institute of Physics, \'{E}cole Polytechnique F\'{e}d\'{e}rale de Lausanne (EPFL), Lausanne, Switzerland}

\author{William R. Clements}
\affiliation{ORCA Computing, London, UK}

\date{\today}

\begin{abstract}
    Though parameter shift rules have drastically improved gradient estimation methods for several types of quantum circuits, leading to improved performance in downstream tasks, so far they have not been transferable to linear optics with single photons.
    In this work, we derive an analytical formula for the gradients in these circuits with respect to phaseshifters via a generalized parameter shift rule, where the number of parameter shifts depends linearly on the total number of photons. Experimentally, this enables access to derivatives in photonic systems without the need for finite difference approximations. Building on this, we propose two strategies through which one can reduce the number of shifts in the expression, and hence reduce the overall sample complexity. Numerically, we show that this generalized parameter-shift rule can converge to the minimum of a cost
function with fewer parameter update steps than alternative techniques. We anticipate that this method will open up new avenues to 
   solving optimization problems with photonic systems, as well as provide new techniques for the experimental characterization and control of linear optical systems.
\end{abstract}

\maketitle

\section{Introduction}
Linear optics with single photons has seen widespread interest, with significant efforts committed to developing architectures allowing efficient computation with photons. Near term applications include boson sampling~\cite{aaronson2011computational}, while in the longer term this framework can enable universal quantum computation~\cite{knill2001scheme}. At sufficiently large scales, these computations are out of reach for classical machines.

Boson sampling is a computational paradigm in which non-classical light sources are interfered within an interferometer coupling multiple optical modes. At the output, a probability distribution is generated, which is strongly believed to be classically hard to replicate~\cite{aaronson2011computational}. 
These architectures are ideal testbeds to show quantum advantage with current state-of-the-art photonic technologies~\cite{zhong2020quantum,madsen2022quantum}. Concurrently, applications to boson sampling have been suggested for several problems, such as optimization~\cite{bradler2021certain}, chemistry~\cite{sparrow2018simulating,huh2015boson,shang2024boson}, graph problems~\cite{mezher2023solving}, and machine learning (ML)~\cite{gan2022fock,wallner2023towards,yin2024experimental,saggio2021experimental}.

Universal linear optical quantum computation (LOQC) is another equally important avenue of research~\cite{knill2001scheme,raussendorf2001oneway,raussendorf2003measurementbased,bombin2021interleaving}, which requires photon-photon interactions to implement arbitrary operations among photon-encoded qubits. In particular, the Knill-Laflamme-Milburn (KLM) scheme~\cite{knill2001scheme} relies on implementing two-qubit gate operations non-deterministically by selecting the events that successfully perform the desired entangling operation. In measurement-based approaches~\cite{raussendorf2001oneway,raussendorf2003measurementbased}, instead, first a large number of \textit{resource} states (known examples are Bell or GHZ states) are created, and subsequently an arbitrary algorithm is run by performing single-qubit measurements.

Concurrently to advances in linear optics, variational quantum computing (VQC) in qubit-based platforms has seen rising interest. Variational algorithms have been demonstrated for a wide range of tasks, and can be suitable for noisy intermediate-scale quantum (NISQ) hardware~\cite{cerezo2021variational}. Typical sought-after applications include ground- and excited-state problems~\cite{peruzzo2014variational,kandala2017hardware,higgott2019variational,jones2019variational}, optimization tasks~\cite{farhi2014quantum,lin2016performance,wang2018quantum}, as well as quantum machine learning (QML) tasks~\cite{biamonte2017quantum,farhi2018classification,schuld2020circuit}.
In such variational tasks, the first step is to find an objective loss function, measurable on a quantum device, that faithfully encodes the problem and its solution. Then, the main routine is to classically update a set of parameters of the quantum device, and optimize them so as to minimize the objective loss function. Optimization is often carried out via gradient-based methods, as these methods often perform well even in complex problems where gradient-free counterparts would fail.

However, VQC in photonic platforms remains relatively unexplored. The combination of the two could bring fresh insight into how these devices could tackle not only chemistry and ML problems, but also characterization and control tasks~\cite{khatri2019quantum,  caro2022outofdistribution, jerbi2023power, Volkoff2021Universal, arrazola2019machine}. As an example, this includes more sophisticated strategies investigating the impact of imperfections in the components of linear optical systems~\cite{pai2019matrix}, which would allow for better understanding of how these imperfections can affect tasks such as resource state generation~\cite{cao2024photonic,chen2024heralded,bartolucci2021creation}.
\begin{figure}
    \centering
    \includegraphics[width=0.8\linewidth]{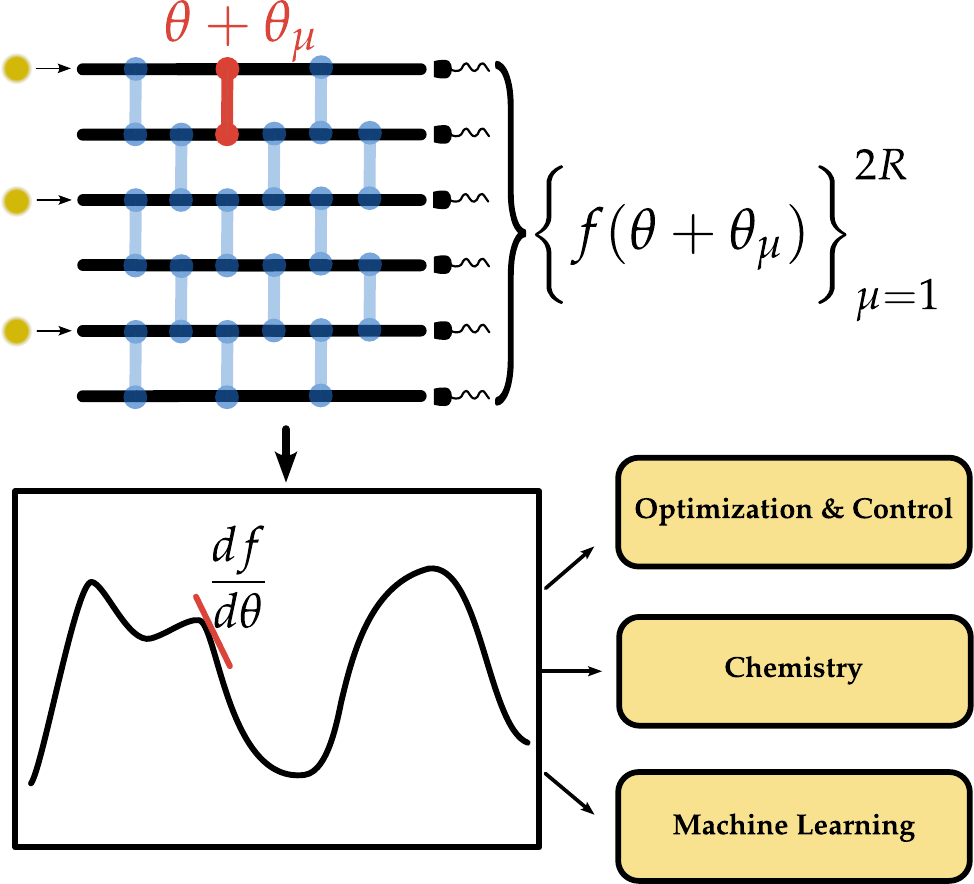}
    \caption{\textbf{The generalized parameter-shift rule in linear optics.} An arbitrary function $f(\cdot)$ generated by a linear optical circuit can be expanded into a finite Fourier series with $R$ positive frequencies. By evaluating its value at $2R$ shifted values of the argument, we are able to reconstruct its derivative, allowing us to more easily tackle a range of tasks.}
    \label{fig:conceptual_figure}
\end{figure}

In qubit-based platforms, as well as in continuous-variable optical systems, efficient strategies to obtain gradients have already been developed~\cite{li2017hybrid,mitarai2018quantum,schuld2019evaluating,banchi2020training,izmaylov2021analytic,kyriienko2021generalized,mari2021estimating,wiersema2024here,Banchi2021measuring,arrasmith2020operator,stokes2020quantum,sweke2020stochastic,kubler2020adaptive,kyriienko2021generalized}. In particular, for qubit-based systems the well-known parameter-shift rule (PSR) has attracted a lot of interest as a way to reconstruct the exact gradient for any class of parametrized quantum circuits that contains generators with symmetric eigenvalues~\cite{li2017hybrid,mitarai2018quantum,schuld2019evaluating}. This was then later extended to generators with a larger eigenspectrum~\cite{wierichs2022general,kyriienko2021generalized}. Similarly, for optical circuits in the continuous-variable regime, observables can be differentiated by performing a similar parameter shift-like rule~\cite{schuld2019evaluating}. These methods generally allow gradients to be estimated with more accuracy than finite difference methods, which are only accurate in the limit of small parameter changes. However, no such derivation exists for linear optical circuits with Fock input states.

In this work, we derive a scheme to estimate derivatives in such devices. In particular, we show that, for an arbitrary Hermitian operator, the dependence of its expectation value on a parameter in the circuit can be expanded as a finite Fourier series. From this, one can leverage trigonometric interpolation in order to reconstruct the function, and its derivatives, by taking expectation values of the same operator at shifted values of the parameter.\\

The paper is structured as follows: in Section~\ref{sec:background} we review the main concepts of linear optics and previous work on PSR approaches to computing gradients. In Section~\ref{sec:gpsr_main_result} we show how to expand an arbitrary expectation value into a finite Fourier series, after which \textit{trigonometric interpolation} is introduced. Together, these two ingredients allow us to derive the main result of this work. In Section~\ref{sec:gpsr_reduced} we propose two strategies to reduce the number of circuit evaluations needed to compute gradients. The first strategy is based on a light cone argument of the parametrized operation. The second strategy is obtained by considering a specific, yet quite general and physically relevant, family of observables. Finally, in Section~\ref{sec:applications} we demonstrate the gradient formula on some example applications. We focus on a QML binary classification task, showing that, given the same number of parameter updates, a linear optical circuit using the generalized parameter-shift rule can train more consistently than other optimization methods. Secondly, we show how the result of this work could be applied to study the impact of parameters on the fidelity of the output two-qubit state in a Bell state generation circuit. 

\section{Background}\label{sec:background}
In this section, we introduce the basics of linear optics. This also allows us to explicitly define the notation that will be used for the rest of this work. Subsequently, we review some of the prior work concerned with the derivation of PSR for qubit-based systems.

\subsection{Linear Optics with Fock states}
In the linear optics framework, passive, particle-number conserving operations are applied to several optical modes. The building blocks of the unitary transformations are phaseshifters and beamsplitters: given two modes with creation operators $a\ad_1,a\ad_2$~\footnote{For ease of notation, we do not use the hat notation for operators, provided there is no ambiguity, in which case we will specify the nature of the object.}, the unitaries of these processes are respectively defined by
\be\label{eq:linear_optics_unitaries}
\begin{split}
\PC_{\phi} &= \exp[i \phi a_1\ad a_1]\,,\\
\UC_{50:50} &=\exp[i\dfrac{\pi}{4}(a_1\ad a_2 + a_2\ad a_1)]\,.
\end{split}
\ee
One can show that such unitary evolutions correspond to the following unitary transformations applied to the mode vector $\vec{a}\ad =(a_1\ad, a_2\ad)^T$
\be\label{eq:BS_PS}
P_{\phi} =
\begin{pmatrix}
   e^{i\phi} & 0 \\
   0 & 1 
\end{pmatrix}\,,\qquad
U_{50:50} =\dfrac{1}{\sqrt{2}}
\begin{pmatrix}
   1 & i \\
   i & 1 
\end{pmatrix}\,.
\ee
Then, beamsplitters $U_{BS}(\vartheta,\phi)$ of tunable transmission (or reflectivity) are straightforwardly implemented via a Mach-Zehnder interferometer (MZI) by noting that the following identity holds
\be\label{eq:tunable_BS}
\begin{split}
    U_{BS}(\vartheta,\phi) &:= U_{50:50}P_{\pi+2\vartheta}U_{50:50}P_{\pi+\phi} \\
    &= e^{i\vartheta}
\begin{pmatrix}
   e^{i\phi}\cos(\vartheta) & \sin(\vartheta) \\
   -e^{i\phi}\sin(\vartheta) & \cos(\vartheta) 
\end{pmatrix}\,.
\end{split}
\ee
For $n$ photons and $m$ modes, a general Fock state is of the form
\be
    \ket{\nv} =\ket{n_1,\dots,n_m}  = \prod_{i=1}^m \dfrac{(a_i\ad)^{n_i}}{\sqrt{n_i!}}\ket{0}^{\otimes m}\,,
\ee
with $\sum_{i=1}^m n_i = n$, and the state $\ket{0}$ indicates the vacuum state in a given mode. Note that the dimension of the Hilbert space $\HC$ spanned by these Fock states is given by $d:=\text{dim}(\HC) = \binom{n+m-1}{n}$. The operations allowed by linear optics, defined in Eq.~\eqref{eq:BS_PS}, induce a unitary transformation $U\in\Ubb(m)$ on the vector of creation operators $\av\ad :=(a_1\ad,\dots,a_m\ad)^T$, namely
\be
    a_i\ad \xmapsto{\hspace{2pt}U\hspace{2pt}} a_i'^{\dagger} =(U^T\av\ad)_i = \sum_{j=1}^m a_j\ad u_{ji}\,.
\ee
where $(\cdot)_i$ picks the $i$-th row after the unitary mode transformation. Note that, henceforth, we will use the notation $\UC\in\Ubb(d)$ to denote a unitary evolution applied to states $\ket{\psi} \in \HC$, and $U\in \Ubb(m)$ to denote the associated unitary mode transformation which is well-defined for linear optics.

\subsection{Variational Quantum Computing}
At the heart of many algorithms and applications in quantum computing lies the ability to parametrize a desired unitary evolution by a set of parameters $\thv=(\theta_1,...,\theta_M)^T$. These  are used to generate a circuit implementing the unitary $\UC(\thv)$. At the output, some quantity of interest $f(\thv)$ is usually measured via the expectation value of a suitably defined observable $O = \sum_k f_k O_k$. This amounts to
\be
f(\thv) = \expval{O} = \bra{\psi_0}\UC\ad(\thv) O \UC(\thv)\ket{\psi_0} \,,
\ee
where $\ket{\psi_0}$ is the initial quantum state. If the task considered contains classical data, $\UC(\thv)$ could also depend on some input classical data (which we omit here for ease of notation).
This notation remains valid both for linear optics, where in this case we would consider both beamsplitters and phaseshifters $\thv=(\vartheta_1,\phi_1,\dots,\vartheta_M,\phi_M)^T$, and for qubit-based systems, where typical choices of unitary are the hardware-efficient ansatz~\cite{kandala2017hardware}, the alternating operator ansatz~\cite{farhi2014quantum,hadfield2019from} and the coupled-cluster ansatz~\cite{bartlett2007coupledcluster,cao2019quantum}.

Quite generally, one can assume $\UC(\thv) = \prod_{k=1}^M \UC_k(\theta_k)$, where $\UC_{k}(\theta_k)$ is again an arbitrary unitary (in the case of linear optics, generated by the unitaries defined in Eq.~\eqref{eq:linear_optics_unitaries}).
In this work, we are interested in the functional dependence with respect to a single parameter $\theta_k$. Thus, henceforth we will highlight the dependence on one parameter only  $\UC(\thv)=\VC\UC_k(\theta_k)\WC$, where the two unitaries $\VC,\WC$ have implicit dependence on all the other parameters. Specifically
\be
\WC = \UC(\theta_{k-1})\dots\UC(\theta_1)\,,\, \VC=\UC(\theta_M)\dots\UC(\theta_{k+1})\,.
\ee
With this notation, we consider $f$ to be a univariate function of $\theta_k$
\be
f(\theta_k) = \bra{\psi_0}\WC\ad \UC\ad_k(\theta_k)\VC\ad O \VC \UC_k(\theta_k)\WC\ket{\psi_0}
\ee
For notational simplicity, we will also rename $\theta_k\rightarrow\theta$.
\subsection{Qubit-based parameter-shift rules}

In qubit-based platforms, we often want to compute the derivative with respect to $\theta$, where the parameter appears in the ansatz  through the unitary gate $e^{i\theta G}$. If the generator $G$ contains two distinct eigenvalues $\{\pm \lambda \}$ such that $G^2 = \lambda^2\Ibb$, its derivative can be fully reconstructed with the well-known two-term PSR~\cite{li2017hybrid,mitarai2018quantum}
\be\label{eq:param_shift}
    f'(\theta) = \dfrac{\lambda}{2\sin{(\lambda a)}}\big(f(\theta+a)-f(\theta-a)\big)\,.
\ee
As an example, the formula in Eq.~\eqref{eq:param_shift} holds for Pauli words, i.e. generators where $G=\bigotimes_{i=1}^n P_i\,,\,P_i\in \frac{1}{2}\{\Ibb,\sigma_x,\sigma_y,\sigma_z\}$ (where $n$ in this example is the number of qubits).
Hence, in practice, by shifting the parameter $\theta$ by some amount $a$ and taking the expectation value of the observable $O$, and similarly with a shift $-a$, one can experimentally recover the derivative of the function.

The formula in Eq.~\eqref{eq:param_shift} can be generalized to more general unitary operations~\cite{wierichs2022general,kyriienko2021generalized}, allowing the generator $G$ to possess a spectrum with more than two symmetric eigenvalues. Here, the number $R$ determining the number of circuit evaluations corresponds to the cardinality $|\Omega^+|$ of the set of distinct positive spectral gaps, i.e. $\Omega^+ := \{\omega= \lambda_i-\lambda_j\text{ : }\omega >0\,,\, \lambda_i,\lambda_j \in \sigma(G)\}$, where $\sigma(G)$ denotes the eigenspectrum of the generator. In this setting, the derivative can be reconstructed by evaluating the function at $2R$ shifted values of the parameter.

While many gradient strategies have now been proposed on qubit-based and continuous-variable platforms~\cite{li2017hybrid,mitarai2018quantum,schuld2019evaluating,banchi2020training,izmaylov2021analytic,kyriienko2021generalized,mari2021estimating,wiersema2024here,Banchi2021measuring,arrasmith2020operator,stokes2020quantum,sweke2020stochastic,kubler2020adaptive}, there is no similar result which allows the differentiation of some objective function within the framework of linear optics with Fock input states. The work presented in Ref.~\cite{defelice2024differentiation} provides a way to differentiate expectation values of arbitrary observables in such systems, but with a considerable overhead, requiring $m+1$ additional modes and one more photon. 

Our work fills this gap: by considering the eigenspectrum of a phaseshifter acting on an arbitrary mode $\nu$, $\nu\in\{1,\dots,m\}$ where the generator is given by the number operator $n_\nu = a\ad_\nu a_\nu$ we find that the function's dependence on $\theta$ can be inferred by knowing its value at $2R+1$ distinct points, and are able to derive the analytical form of its derivatives.

\section{Gradients for linear optics with Fock states}\label{sec:gpsr_main_result}
In this section, we derive our main result. First, considering a general parametrized quantum circuit (PQC) constructed with linear optical components, and concentrating on the dependence of a parameter $\theta$, we show that the expectation value of an arbitrary observable with a Fock input state can be expanded into a finite Fourier series. Secondly, we introduce the concept of \textit{trigonometric interpolation}, which allows us to fully reconstruct a finite Fourier series function with $R$ positive frequencies, having knowledge of the value of the function at $2R+1$ points. Combining the two results, we can then reconstruct the expectation value's functional form in $\theta$, and straightforwardly deduce its derivative. 
In the limit of infinite sampling, this gives us access to the exact derivatives. For finite sampling, the gradient will not be exact but the method can still be more accurate than the finite difference method, which is only exact in the limit of infinitesimally small parameter changes. 
\begin{figure}
\centering
\includegraphics[width=0.7\linewidth]{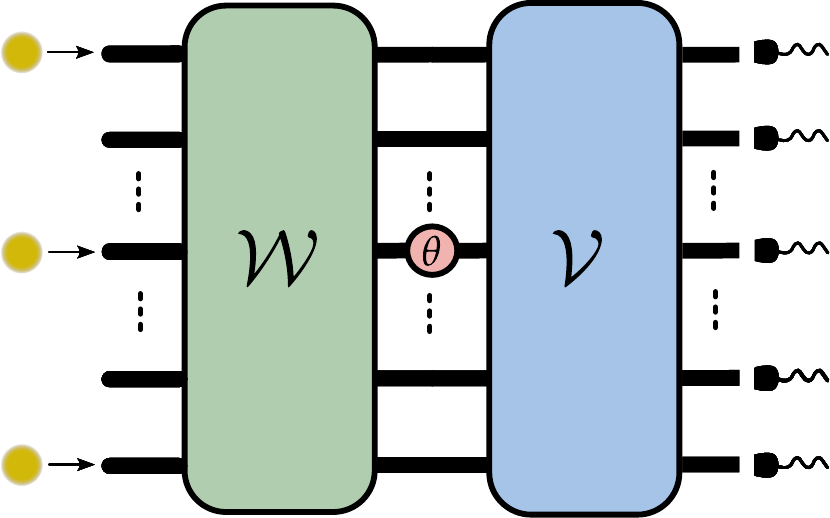}
\caption{\textbf{Parametrized interferometer.} The parametrized quantum circuit that we consider for the derivation of a closed-form derivative formula. We highlight the dependence on a phaseshifter, the parameter with respect to which we wish to compute the derivative. This can be straightforwardly generalized to apply also to beamsplitter parameters by considering that a tunable beamsplitter can be decomposed into two beam splitters and one phase shifter.}\label{fig:fockcircuit}
\end{figure}

\subsection{Fourier expansion of an expectation value}\label{sec:fourier_expansion}
It has been shown~\cite{gan2022fock} that for a linear optical circuit comprised of trainable blocks, any observable can be expanded into a finite Fourier series as a function of its programmable parameters, in a manner similar to qubit-based systems~\cite{perez2020data,schuld2021effect}. 

Our work follows a similar derivation. Consider a linear optical circuit such as the one in Fig.~\ref{fig:fockcircuit}, where arbitrary unitaries $\WC, \VC$ are applied before and after a parametrized phaseshifter $\PC_{\theta}$. Then, for any hermitian observable $O$ and initial state $\ket{\nv}$, with total photons $n$, we have
\begin{align}
    f(\theta) :=&\fexpval{\nv}{\UC\ad O \UC}{\nv}\label{eq:expval_firstline} \\
    =& \fexpval{\nv}{\WC\ad \PC_{\theta}\ad \VC\ad O \VC \PC_{\theta} \WC}{\nv}\label{eq:expval_secondline} \\ 
    =& \sum_{\sv,\sv'} O_{\sv,\sv'} \fexpval{\nv}{\WC\ad \PC_{\theta}\ad \VC\ad}{\sv} \fexpval{\sv'}{\VC \PC_{\theta} \WC}{\nv} \label{eq:expval_thirdline}\,,
\end{align}
where in Eq.~\eqref{eq:expval_thirdline} we inserted two identities on both sides of the observable. The sum is over the basis of the Hilbert space $\HC$, namely $\{\ket{\sv}=\ket{s_1,\dots,s_m}\,|\,\sum_i^m s_i = s\}$. Note that, for an arbitrary matrix $\mathcal{M}$ we use the notation $\mathcal{M}_{\sv,\nv}:=\fexpval{\sv}{\mathcal{M}}{\nv}$.

We can then expand the amplitudes. Let $\nu$ denote the mode that the $\theta$ phaseshifter acts on, $\hat{n}_\nu=\hat{a}\ad_\nu \hat{a}_\nu$ denote the corresponding number operator and $t_\nu$ be the eigenvalue associated to state $\ket{\vec{t}}$.  We then have $\PC_{\theta}\ket{\tv}=e^{i\hat{n}_\nu \theta}\ket{\tv}=e^{it_\nu \theta}\ket{\tv}$.
Thus we find 
\be\label{eq:expandamp}
    \fexpval{\sv'}{\VC \PC_{\theta} \WC}{\nv} = \sum_{\tv} \VC_{\sv',\tv} e^{it_{\nu} \theta} \WC_{\tv,\nv}\,,
\ee
Finally, the expectation value takes the following form
\be
f(\theta) = \sum_{\tv,\tv'}c_{\tv,\tv'}e^{i(t_\nu-t_\nu')\theta} \,,
\ee
where
\be
    c_{\tv,\tv'} = \sum_{\sv,\sv'} \VC_{\sv',\tv}\WC_{\tv,\nv} O_{\sv,\sv'} \WC^*_{\tv',\nv} \VC^*_{\sv,\tv'}\,.
\ee
It can be further simplified when accounting for the degeneracy of each frequency $\omega =t_\nu-t_\nu'$ to give
\be\label{eq:exp_val_fourier}
\begin{split}
f(\theta) &= \sum_{\omega \in \Omega} c_{\omega}e^{i\omega \theta}\hspace{10pt} \\
\text{where}\hspace{10pt}c_{\omega} &= \sum_{\tv,\tv'\text{ s.t. } t_\nu-t_\nu'=\omega} c_{\tv,\tv'}\,.
\end{split}
\ee
In this particular case the Fourier spectrum is $\Omega=\{-n,\dots,0,\dots,n\}$. The function $f(\theta)$ can also be cast into a trigonometric polynomial of order $n$
\be
    f(\theta) = a_0 + \sum_{\omega \in \Omega^+} a_\omega \cos{(\omega \theta)} + b_{\omega}\sin{(\omega \theta)}\,,
\ee
where we denote $\Omega^+ = \{1,\dots,n\}$ the set of strictly positive frequencies out of $\Omega$ and the coefficients $a_{\omega} = (c_{\omega}+c_{-\omega})$, $b_{\omega} = i(c_{\omega}-c_{-\omega})$. Note that the $c_{\omega}$'s depend on $\VC,\WC$ and hence on the remaining parameters in the circuit. Having expressed $f(\theta)$ as a finite Fourier series, we now proceed to show how the dependence of $f$ on $\theta$ can be inferred via trigonometric interpolation.

\subsection{Trigonometric Interpolation}\label{sec:trig_interpolation}
In this section we introduce the concept of \textit{trigonometric interpolation}, i.e. the problem of reconstructing a function which can be expressed as a finite Fourier series. Suppose the goal is to try and reconstruct a function which, similarly to Eq.~\eqref{eq:exp_val_fourier}, is in the following form
\be\label{eq:fourier_func}
    f(\theta) = \sum_{k=-R}^R c_k e^{ik\theta} \,,
\ee
where $R$ represents the number of positive frequencies. Given access to a set of points $\{(\theta_\mu,f(\theta_\mu)), \mu=1,\dots,K\}$, we recall the result that there exists an exact and unique solution for $K=2R+1$. In particular
\begin{lemma}[Trigonometric Interpolation, Ref.~\cite{atkinson1991introduction}]\label{lemma:trig_interp}
Suppose a function of the form as in Eq.~\eqref{eq:fourier_func}. Suppose furthermore that we are given the values of $f$ at $2R+1$ equidistant points in a $2\pi$ interval, e.g. we have knowledge of $\{(\theta_\mu,f(\theta_\mu))\, |\, \theta_\mu=2\pi \mu/(2R+1),\, \mu=-R,\dots,R\}$. Then one can show that
\be\label{eq:fourier_coeff}
    c_k = \dfrac{1}{2R+1}\sum_{\mu=-R}^{R} e^{-ik\theta_\mu}f(\theta_\mu)\hspace{2pt} \forall\, \mu=-R,\dots,R
\ee
\end{lemma}
A proof is provided in Appendix ~\ref{appendix:trig_interpolation}. For the interested reader we also refer to Refs.~\cite{atkinson1991introduction,zygmund2003trigonometric}. We remark that, since the $c_k$'s depend on all the other parameters in $\VC,\WC$, the solution to the interpolation problem provided in Lemma~\ref{lemma:trig_interp} is true only for a fixed, given value of them. Note that plugging the result of Eq.~\eqref{eq:fourier_coeff} into Eq.~\eqref{eq:fourier_func}, one obtains that $f(\theta)$ can be expressed as a sum of the so-called \textit{Dirichlet kernels}
\be\label{eq:def_trig_interpolation}
    f(\theta) = \sum_{\mu=-R}^{R} f(\theta_\mu) D(\theta-\theta_\mu) \,,
\ee
defined as
\be\label{eq:def_dirichlet_kernel}
D(\theta) := \dfrac{\sin{\big((2R+1)\theta/2\big)}}{(2R+1)\sin{\big(\theta/2\big)}} \,.
\ee
Hence, a function with $2R$ symmetric frequencies, or equivalently $R$ positive frequencies, can be expanded as a sum of $2R+1$ Dirichlet kernels, each one of them weighted by the function evaluated at the equidistant points $\theta_\mu$.

\subsection{Generalized parameter-shift rule}\label{sec:GPSR}

Combining the results of Section~\ref{sec:fourier_expansion}, \ref{sec:trig_interpolation} we can now present our expression for computing the derivatives of expectation values of parametrized linear optical circuits.
\begin{theorem}[Generalized parameter-shift rule in linear optics]\label{thm:GPSR} Given an $m$-mode  PQC with $n$ photons generating a parametrized unitary evolution $\UC(\thv)\in \Ubb(d)$, the $k$-th order derivative of a function as defined in Eq.~\eqref{eq:expval_firstline} with respect to an arbitrary parameter $\theta$ is  given by
\be\label{eq:gpsr}
f^{(k)}(\theta) = \sum_{\mu=-n}^n f(\theta + \theta_\mu)D^{(k)}(-\theta_\mu) \,,
\ee
where $\theta_\mu = \frac{2\pi\mu}{2n+1}$ and the superscript $(k)$ indicates the $k$-th order derivative.
\end{theorem}
\begin{proof}
    The results of section~\ref{sec:fourier_expansion} show that the expectation value of an arbitrary observable corresponds to a univariate function of an arbitrary phaseshifter parameter $\theta$ in the linear optical circuit, i.e. it has the form of Eq.~\eqref{eq:exp_val_fourier}. Then, applying the result of Lemma~\ref{lemma:trig_interp}, from Eq.~\eqref{eq:def_trig_interpolation} we get that the function's dependence on $\theta$ can be inferred via $2n+1$ circuit evaluations at shifted parameter values
    \be
    f(\theta) = \sum_{\mu=-n}^n f(\theta_\mu)D(\theta-\theta_\mu) \,.
    \ee
    From here, the derivatives of $f(\theta)$ can be inferred
    \be
    \begin{split}
        f^{(k)}(\theta) &= \sum_{\mu=-n}^n f(\theta_\mu) D^{(k)}(\theta-\theta_\mu) \\
        &=\sum_{\mu=-n}^n f(\theta+\theta_\mu) D^{(k)}(-\theta_\mu) \,,
    \end{split}
    \ee
    where in the second equality we shifted the shifted parameters $\theta_\mu$ by $\theta$. Thus we obtain Theorem~\ref{thm:GPSR}.
    \end{proof}

As an example, we explicitly evaluate the first order derivative of the Dirichlet kernel
    \be
        D'(-\theta_\mu) = \dfrac{(-1)^{\mu+1}}{2\sin(\theta_\mu /2)} \,,
    \ee
So that the first-order derivative of $f(\theta)$ is given by
\be\label{eq:GPSR_firstorder}
f'(\theta) = \sum_{\mu=1}^n\big(f(\theta + \theta_\mu)-f(\theta-\theta_\mu)\big)\dfrac{(-1)^{\mu+1}}{2\sin(\theta_{\mu}/2)} \,,
\ee
which more closely resembles a generalization of a parameter-shift-like rule, where we shift the parameter  $\theta$ positively by some amount $\theta_\mu$ and then negatively by the same amount.

In general, this generalized parameter-shift rule (GPSR) needs knowledge of $2n+1$ distinct evaluations of $f(\cdot)$. We remark that while the result is formulated for a parametrized phaseshifter, the same holds true for a parametrized beamsplitter. In fact, given the MZI decomposition in Eq.~\eqref{eq:tunable_BS}, applying a beamsplitter with parameter $\theta$ amounts to applying a phaseshifter with angle $\pi+2\theta$. Thus, after rescaling, the result of Theorem~\ref{thm:GPSR} can be readily applied. 

The result of Theorem~\ref{thm:GPSR} allows us to reconstruct arbitrary order derivatives, but in some tasks we may only be interested in evaluating the derivative $f'$. 
While a closed-form formula is already given in Eq.~\eqref{eq:GPSR_firstorder}, there exists a more efficient way of obtaining first-order derivatives, that was similarly derived for qubit-based systems in Ref.~\cite{wierichs2022general}. By \textit{efficient}, we mean a formula with a corresponding estimator that requires a smaller number of experimental samples to evaluate $f'$ within accuracy $\epsilon$ (see Section~\ref{sec:shot_noise} or Appendix~\ref{appendix:statistical_estimation} for further details). 
The result is achieved by similarly considering the task of interpolating the same Fourier series, but with $2n$ points only via the \textit{modified Dirichlet} kernels. While to reconstruct the entire function they won't be sufficient, they will instead be enough to completely determine the first derivative. As a consequence, we have that

\begin{corollary}[First-order derivative in linear optics]\label{corollary:GPSR_odd}
Given a $m$-mode PQC with $n$ photons generating a parametrized unitary evolution $\UC(\thv)\in \Ubb(d)$, the first order derivative of a function as defined in Eq.~\eqref{eq:expval_firstline} with respect to an arbitrary parameter $\theta$ is  given by
\be\label{eq:gpsr_odd}
f'(\theta) = \sum_{\mu=1}^{2n} f(\theta + \theta_\mu)\dfrac{(-1)^{\mu+1}}{4n\sin^2(\theta_\mu/2)} \,,
\ee
where in this case $\theta_\mu = \frac{(2\mu-1)\pi}{2n}$. 
\end{corollary}
A proof of the Corollary is provided in Appendix~\ref{appendix:gpsr}.
In Appendix~\ref{appendix:lossy_GPSR}, we generalize the main results of both Theorem~\ref{thm:GPSR} and Corollary~\ref{corollary:GPSR_odd} to more realistic, lossy setups. We show that similarly to the ideal scenario, states subject to loss in linear optical systems still allow for a Fourier expansion which contains at most $n$ positive frequencies. Hence, the same results can also be applied to realistic experimental scenarios.

\subsection{Shot noise requirements and accuracy}\label{sec:shot_noise}
\begin{figure}
    \centering
    \includegraphics[width=\linewidth]{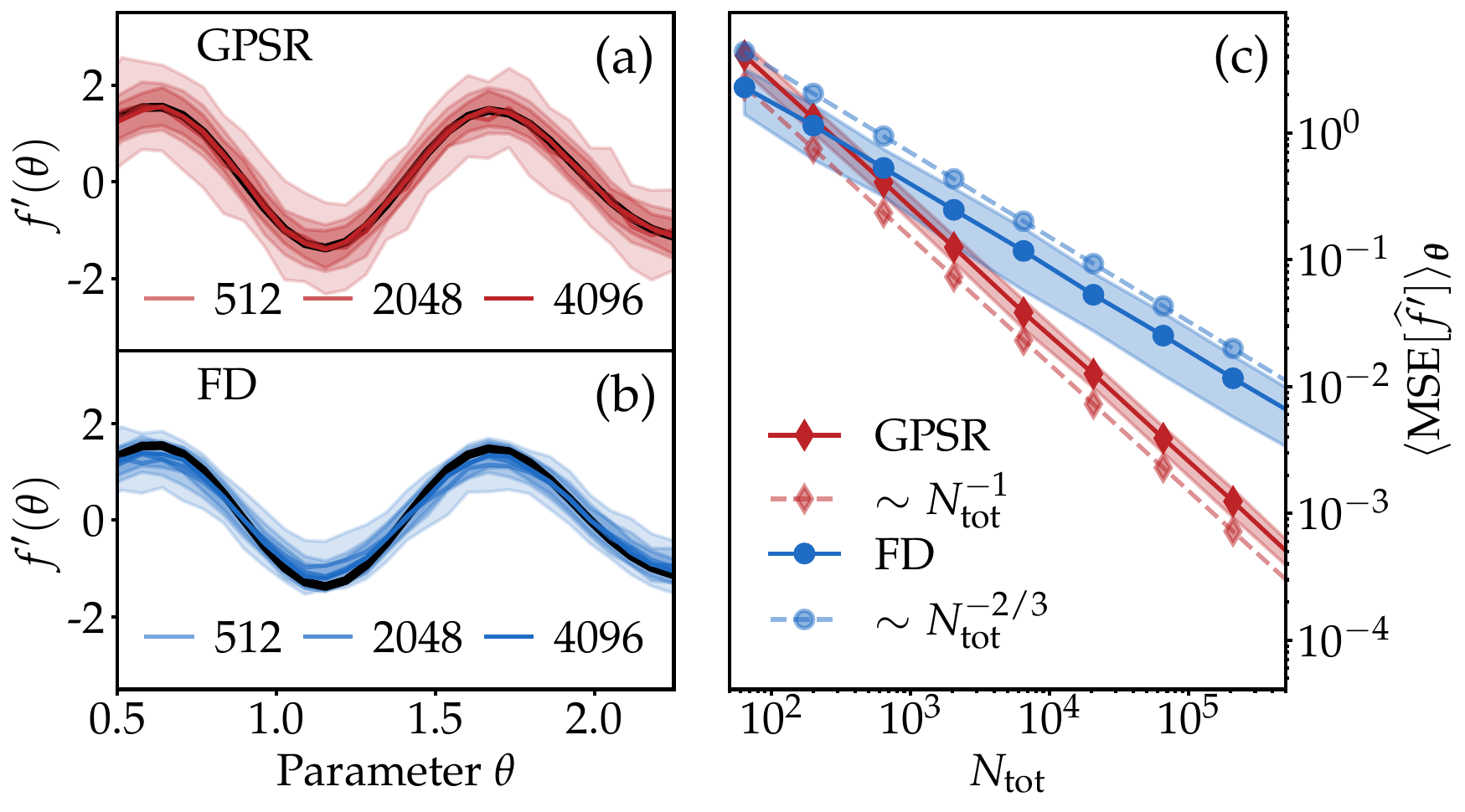}
    \caption{\textbf{First-order derivative of an observable.} We plot the derivative of the expectation value for an observable with eigenvalues chosen uniformly at random in the interval $[-5,5]$ (Eq.~\eqref{eq:random_observable}) in a $4$-mode, $4$-photon system. In particular, (a) contains the derivative reconstructed with the GPSR from Eq.~\eqref{eq:gpsr_odd}, for three different number of samples $\Ntot$, while (b) contains the same experiments with a FD approach. Both plots contain the exact derivative, highlighted in black. In (c), we plot the mean squared error, averaged over multiple sets of parameters, as a function of $\Ntot$ for the two derivatives. We include the theoretical scalings (dotted lines) derived in Appendix~\ref{appendix:statistical_estimation}. Lower number of samples yield a greater statistical uncertainty for the GPSR, but this quickly improves as $\Ntot$ grows and FD remains a biased estimator.}
    \label{fig:gradients_comparison}
\end{figure}
In this section, we evaluate the number of shots required by the GPSR and finite-difference (FD) to evaluate first-order derivatives. 
We perform such an analysis both with Dirichlet kernels and the modified Dirichlet kernels, the latter introduced to achieve the result of Corollary~\ref{corollary:GPSR_odd}.
We compare these to the central FD
\be
f'(\theta) \simeq \dfrac{f(\theta+h)-f(\theta-h)}{2h} \,,
\ee
where the step size $h$ is assumed to be small for the approximate equality to hold. The details for our calculations are provided in Appendix~\ref{appendix:statistical_estimation}.

Across all scenarios, an important assumption is that the physical variance depends weakly on the parameters
\be\label{eq:physical_variance}
    \sigma^2_{\thv} := \langle O^2 \rangle_{\thv} - \langle O \rangle_{\thv}^2 \approx \sigma^2 \,,
\ee
where $\expval{\cdot}_{\thv}=\bra{\nv}\UC\ad(\thv)\cdot\UC(\thv)\ket{\nv}$. If this were not the case, we could think of $\sigma^2$ as the value that maximizes Eq.~\eqref{eq:physical_variance}.

Then, for the GPSR, if we equally distribute the total number of samples $\Ntot$ among each distinct function evaluation in the sum of Eq.~\eqref{eq:gpsr}, 
\be
\Ntot \in \OC\bigg( \dfrac{2\sigma^2 n^2(n+1)}{3\epsilon^2}\bigg) \sim \OC\left(\epsilon^{-2}\right) \,,
\ee
measurements suffice for a desired additive error $\epsilon$. If we distribute $\Ntot$ according to the 1-norm of the weight of the Dirichlet kernel derivative, the number of samples approximately scales as
\be
\Ntot\in\OC\bigg(\dfrac{\sigma^2 n^2 \ln^2(n)}{\epsilon^2}\bigg) \sim \OC\left(\epsilon^{-2}\right)\,.
\ee
Considering instead the GPSR with the modified Dirichlet kernel from Corollary~\ref{corollary:GPSR_odd}, the sample complexity  with 1-norm weighting is further improved to
\be
    \Ntot \in \OC\bigg(\dfrac{\sigma^2 n^2}{\epsilon^2}\bigg) \sim \OC\left(\epsilon^{-2}\right)\,.
\ee
For the FD approach, instead, the number of samples that suffice will scale as~\cite{mari2021estimating}
\be
N_\text{tot} \in \OC\bigg( \dfrac{\sqrt{3}|f^{'''}(\theta)|\sigma^2}{4\epsilon^3}\bigg) \sim \OC\left(\epsilon^{-3}\right) \,,
\ee
where we remark that the scaling is found by minimizing the mean-squared error with respect to the step size $h$. In doing so we find that $h_{\text{opt}} \propto (\beta\Ntot)^{-1/6}$, where similarly to $\Ntot$, $\beta$ is a constant that depends on the specific circuit and observable considered. Hence, in general, it will be difficult to find an optimal step size so that the FD approach is well-behaved.

As an example, we compare in Fig.~\ref{fig:gradients_comparison} the partial derivative of an arbitrarily defined function and the derivative computed with GPSR and FD, for a $4$-mode, $4$-photon system. In the plot on the right-hand side we also plot the mean-squared error of each of the estimators against total number of samples $\Ntot$, along with the theoretical sample complexity scalings. Specifically, we take the definition of Eq.~\eqref{eq:expval_firstline} where $O$ is defined as
\be\label{eq:random_observable}
O = \sum_{\nv}\lambda_{\nv}\ket{\nv}\bra{\nv}\,,
\ee
where the eigenvalues are some fixed values chosen as $\lambda_{\nv}\thicksim U(-5,5)$. We remark that while at very low number of samples the GPSR yields a higher inaccuracy given by the need to evaluate a greater number of distinct circuits, as  $\Ntot$ is increased the GPSR improves over a conventional FD approach, testament to the different scalings of $\Ntot$ with $\epsilon$. Furthermore, while the FD needs heuristic fine-tuning depending on the task considered, the GPSR is already optimal.

\section{Exact simplifications to the GPSR}\label{sec:gpsr_reduced}
In general, the GPSR introduced in section \ref{sec:GPSR} contains a number of circuit evaluations that scales linearly with the total number of photons, and may be asymptotically unfavourable for large enough systems. To circumvent this, in this section we provide two strategies that allow us to reduce the overall sample complexity. In particular, both of them show how we can reconstruct exact derivatives with fewer parameter shifts via a more careful analysis of the Fourier spectrum. This then allows us to reduce the cardinality $R$, which indicates how many distinct circuit configurations need estimation for the GPSR.

\subsection{Number of shifts depends on phaseshifter’s light cone}\label{sec:light_cone}
In this section, we show that depending on the position of the parameter with respect to which we want to compute the gradient, it is possible to reduce the number of shifts in the GPSR, and hence reduce the overall sample complexity. We formalize this by considering the size of the light cone between the initial state and the phaseshifter $\PC_\theta$. By \textit{light cone}, we mean the effective number of modes with which mode $\nu$ interacted before the phaseshifter is applied. In particular, if $\WC\ad$ connects mode $\nu$ to a subset $A$ of the total modes with size $m_A$ and $n_A$ photons, then we are able to reduce the number of parameter shifts to $2n_A$. A drawing of the setup for three different interferometer schemes is provided in Fig.~\ref{fig:resource_savings}, where for each one we draw two light cones generated by the gates situated at their vertex. We then have the following result
\begin{corollary}[Light cone bound on the number of shifts]\label{corollary:connectivity_reduced_GPSR}
    Given the same assumptions of Theorem~\ref{thm:GPSR}, and additionally supposing that the size of the light cone generated by $\WC\ad$ is $m_A$ and contains $n_A$ photons. We then have that the GPSR is reduced to
\be
f^{(k)}(\theta) = \sum_{\mu=-n_A}^{n_A}f(\theta + \theta_\mu)D^{(k)}(-\theta_\mu) \,,
\ee
where $\theta_\mu = \frac{2\pi\mu}{2n_A+1}$. Similarly, the result of Corollary~\ref{corollary:GPSR_odd} reduces to
\be
f'(\theta) = \sum_{\mu=1}^{2n_A} f(\theta + \theta_\mu)\dfrac{(-1)^{\mu+1}}{4R\sin^2(\theta_\mu/2)} \,,
\ee
where $\theta_\mu = \frac{(2\mu-1)\pi}{2n_A}$.
\end{corollary}

A proof of the Corollary is provided in Appendix~\ref{appendix:gpsr_reduced1}. For certain circuits, this can be significant. As an example, we show that for different schemes of parametrized interferometers, a considerable overhead reduction can be achieved when computing the gradient $\vec{\nabla}_{\thv}f(\thv)$.
\begin{figure*}[ht]
    \centering
    \includegraphics[width=\textwidth]{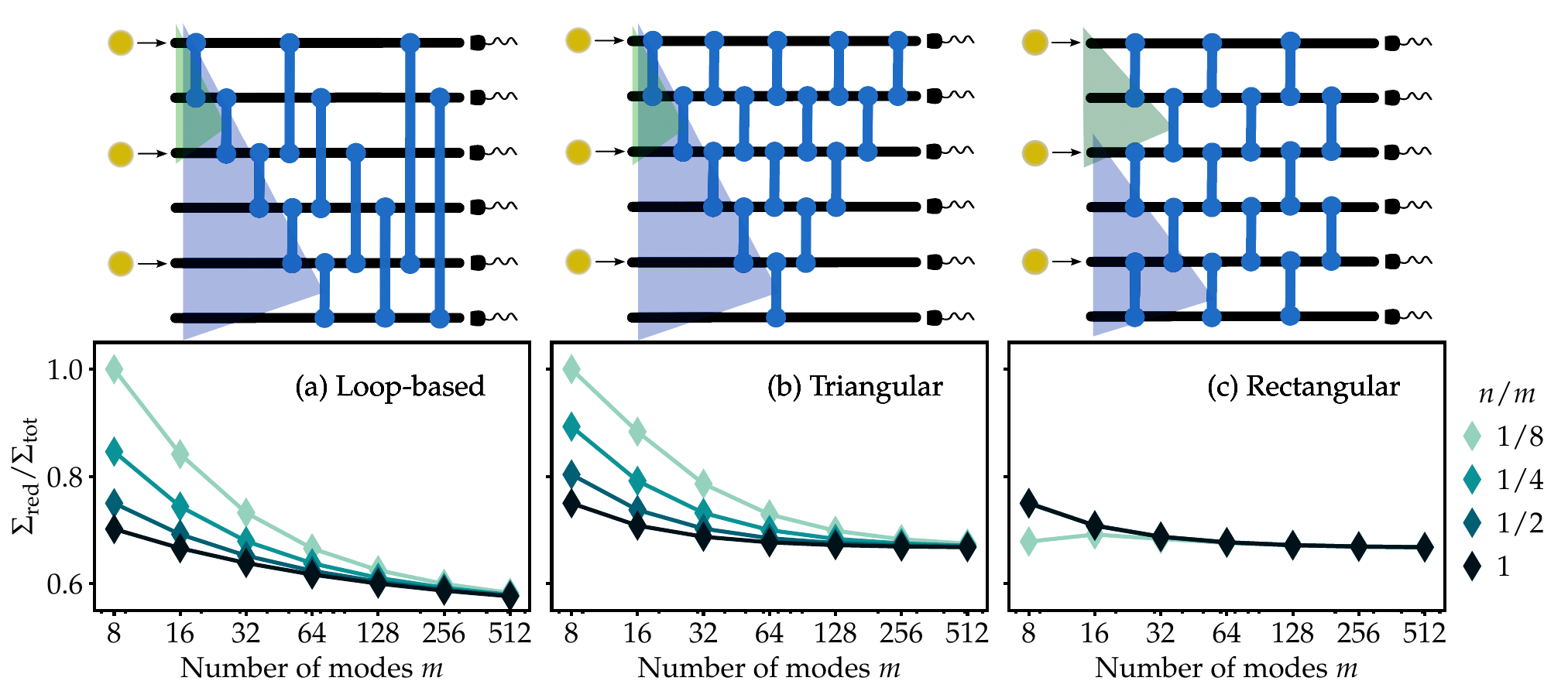}
    \caption{\textbf{Light cone bound resource savings.} We provide a sketch of each of the interferometer schemes with $n/m=1/2$ for which we analyze the reduction in the number of parameters shifts. The circles at the input indicate photons, whereas the vertical blue lines indicate tunable beamsplitter elements (MZIs). For two elements in the circuit, we draw the light cones (green, blue triangle) spanning from the action of the preceding gates. In the plots below, we show the total number of parameter shifts $\Sigma_{\text{red}}$ that one needs in order to get all the gradients following the result of Corollary~\ref{corollary:connectivity_reduced_GPSR}, as a fraction of the total number of shifts $\Sigma_{\text{tot}}$ that one would do as per Theorem~\ref{thm:GPSR} for (a) loop-based scheme (b) triangular scheme and (c) rectangular scheme interferometers.}\label{fig:resource_savings}
\end{figure*}

To analyze the savings in more detail, we focus on the following two figures of merit
\be\label{eq:total_parameter_shifts}
    \Sigma_{\text{tot}} = \sum_{\theta \in \thv} 2n = 2nM \qquad 
    \Sigma_{\text{red}} = \sum_{\theta \in \thv} 2n_A(\theta)\,,
\ee
where we recall that $M$ is the number of parameters. Hence, $\Sigma_{\text{tot}}$ is the total number of parameter shifts (and hence expectation values) that, following the result of Theorem~\ref{thm:GPSR}, we would need to perform in order to compute the gradient. The quantity $\Sigma_{\text{red}}$ represents instead the total number of shifts that we actually require considering Corollary~\ref{corollary:connectivity_reduced_GPSR}. These quantities allow us to quantify the number of circuit evaluations that, in practice, are superfluous when evaluating the gradient of a parametrized interferometer.

In Fig.~\ref{fig:resource_savings} we consider three different types of linear optical circuits with an input state consisting of alternating single photons and vacuum. We consider a time-bin interference scheme with optical loops of different lengths, a triangular scheme~\cite{reck1994experimental} and a rectangular scheme~\cite{clements2016optimal}. For the loop-based scheme, the goal is to analyze shallow-yet-hard-to-simulate configurations~\cite{madsen2022quantum,go2024computational}. In particular, we take inspiration from Ref.~\cite{go2024computational} where their proposed circuit can be described by a sequence of loops of length $[2^0,2^1,2^2,\dots,2^{\lceil\log_2(m)\rceil-1}]$. We assume that each beamsplitter in these schemes is implemented by a MZI containing two programmable phaseshifters, and we are interested in estimating the gradients with respect to all the circuit parameters. A diagram for each of the three circuits is provided, and below we plot the ratio between $\Sigma_{\text{red}}$ and $\Sigma_{\text{tot}}$ as a function of the number of modes, for different levels of photon occupation number $n/m$. 

In all three scenarios we find that, as the number of modes (and photons) grows, the ratios asymptotically approach constant values significantly smaller than $1$. Savings of $30$-$40\%$ can be achieved across the different architectures. In addition, we find that the photon occupation number impacts the ratio only at small system sizes, while for larger systems they converge to the same constant value. In particular, loop-based and triangular scheme have significantly different behaviours across the occupancy levels given their imbalanced structure. The rectangular scheme, on the contrary, has similar behaviour as it acts more uniformly across the modes.

\subsection{Number of shifts depends on the observable}
In this section, we show that for a specific yet physically relevant family of observables, the number of shifts required in the GPSR can be reduced. Analogously to the derivation of the GPSR, we consider the system to undergo some unitary evolutions $\WC,\VC \in \Ubb(d)$, with an additional phaseshifter acting on mode $\nu$ in between the two unitaries. However, if the observable is some polynomial in the number operators $(\hat{n}_1,\dots,\hat{n}_m)$ with degree $p<n$, we show that the number of positive frequencies in the Fourier expansion is given by $p$. Hence, $2p$ evaluations of the function will be required to reconstruct its exact derivative.

More concretely, let us suppose that the observable is a $p$-degree monomial in the photon number operators across the modes,
\be\label{eq:observable}
    O := \mathfrak{f}(\nv,p) = n_1^{p_1}\dots n_m^{p_m}\qquad \sum_i p_i = p \,, p_i\in \Nbb \,.
\ee
We call this a \textit{number-ordered} observable. Usually, we consider $p_i\in \{0,1\}$, so that if only one of the $p_i$'s is non-zero, then the observable amounts to an average photon number measurement, and if two $p_i$'s are non-zero it is instead photon number correlation measurement between two modes. However, we will not make any assumption other than what is stated in Eq.~\eqref{eq:observable}. Furthermore, we only consider monomials since by linearity, generalization to a polynomial is straightforward as long as we consider the monomial with the highest degree.

We stress that such observables are commonly encountered. Indeed, also observables with \textit{normal} ordering such as average photon number, second order correlation and more generally $p$-th order correlation, which are often of interest both theoretically  and experimentally~\cite{glauber1963quantum,hanbury1979test,grunwald2020nonquantum}, can be recast as linear combinations of number-ordered observables.
In particular, one can show that observables defined with the same number of creation and annihilation operators, but with different ordering, will similarly be number-ordered polynomials with same order $p$. As an example, in Appendix~\ref{appendix:gpsr_reduced2} we show that number-ordered monomials of degree $p$ are expressed as normal-ordered polynomials of the same degree. Similarly, we could express an arbitrary normal-ordered operator as a sum of number-ordered ones. In the same way, the highest degree would match for both expressions. 

First, let us make an observation that will be useful later
\begin{observation}\label{obs:one_degree_unitary}
    Consider the unitary mode transformation $U=VP_{\theta}W \in \Ubb(m)$. Then, each entry of $U$,  $U\ad$ are 1-degree polynomials in $e^{i\theta}$, $e^{-i\theta}$ respectively.
\end{observation}
This is because the unitary mode transformation of the phaseshifter is $P_\theta = \text{diag}(1,\dots,1,e^{i\theta},1,\dots,1)$ where the exponential term appears at the $\nu$-th entry.
We are now in the position to state the result which upper bounds the number of frequencies in the observable

\begin{corollary}[Observable-dependent frequency spectrum]\label{corollary:observable_frequency}
    Consider an initial state $\ket{\nv}$, a unitary mode transformation given by $U=VP_{\theta}W$ and an observable as defined in Eq.~\eqref{eq:observable}. Then, the number of frequencies of the expectation value as a function of the parameter $\theta$ will be given by
    \be
    R = \mathrm{min}\{p,n\}\\,
    \ee
    where we recall that $n$ is the total number of photons, and $p$ is the overall degree of the observable in powers of average photon number operators, each mode contributing with $p_i$ and $\sum_i p_i = p$.
\end{corollary}
 A detailed proof is provided in Appendix~\ref{appendix:gpsr_reduced2}. The main idea behind it is realizing that the particular form of the observable makes it so that, in the expectation value, there will be a $p$-fold product of the unitary mode transformation elements. We do so by applying the unitary in a Heisenberg fashion directly to the mode operators defining the observable. By invoking Observation~\ref{obs:one_degree_unitary}, we can investigate how many frequencies we are left with, and are able to complete the proof. To illustrate the result of Corollary~\ref{corollary:observable_frequency} we provide numerics in Fig.~\ref{fig:fourier_freqs} showing the Fourier spectrum for specific examples of observables.
 
 This result can be practically very impactful. If, say, a certain task requires considering observables which contain second-order correlations between modes, for example
 \be
    \sum_{i\leq j}^m o_{ij}\hat{a}_i\ad \hat{a}_j\ad \hat{a}_i \hat{a}_j \,,
 \ee
 then Corollary~\ref{corollary:observable_frequency} ensures that the number of parameter shifts need not to scale with the number of photons, but instead four shifts will suffice. Finally, in Appendix~\ref{corollary:observable_frequency} we generalize the result to operators which, for a given mode $i$, do not contain the same number of creation and annihilation operators, e.g. $\sim (a\ad_i)^{q_i} a_i^{r_i}$. This is relevant for hermitian operators of the form
 \be
f\prod_{i=1}^m(a_i\ad)^{q_i} a^{r_i}_i +  f^*\prod_{i=1}^m(a_i\ad)^{r_i} a^{q_i}_i\,,\,f\in \Cbb \,,
\ee
i.e. operators where each term is not hermitian, but the sum is.
\begin{figure}
     \centering
     \includegraphics[width=\linewidth]{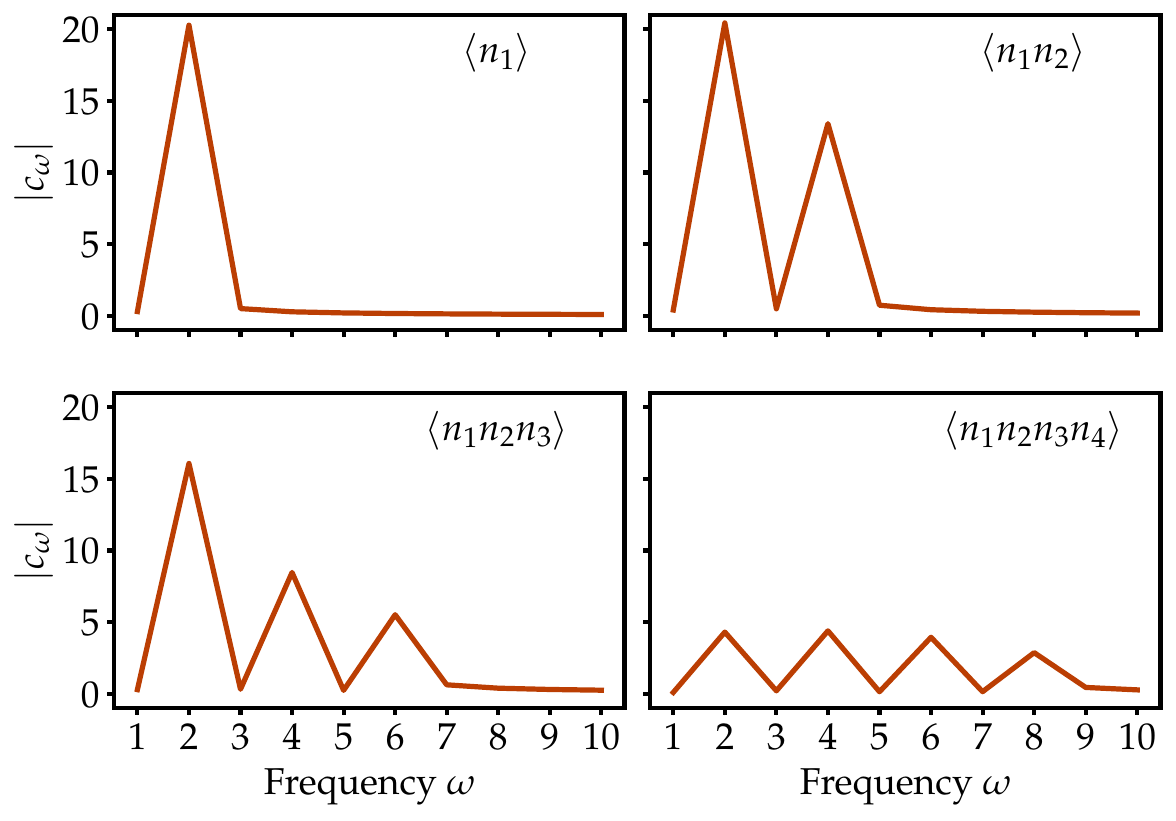}
     \caption{\textbf{Fourier spectrum of average photon number and multi-mode correlations.} For multiple observables, we plot the absolute 
    value of the coefficients $c_{\omega}$ (introduced in Section~\ref{sec:fourier_expansion}), each indicating the Fourier coefficient of the frequency $\omega$, for a linear optical circuit with input state $\ket{2,1,1,1}$. As expected, the number of non-zero coefficients is not determined by the total photon number, but rather the degree of each observable as a polynomial in the number operators.}
     \label{fig:fourier_freqs}
 \end{figure}

\section{Applications}\label{sec:applications}
In this section, we demonstrate two applications where the GPSR is employed and compared against alternative approaches to computing gradients. In the first example, we focus on a QML task that can be run in a linear optical setup, and compare using the GPSR against FD methods to calculate gradients. We show that using the GPSR can result in more reliable training compared to FD methods. In the second example, we instead show how the GPSR could be used to accurately study a Bell state generation circuit.\\

\subsection{QML classification}
As a first example, similarly to Ref.~\cite{gan2022fock}, we consider a simple binary classification task. Given some training data $\mathcal{X}=\{(\xv_i,y_i)\}_{i=1}^{N_{\text{tr}}}$, where $\xv_i\in\Rbb^q$ represents the data features and $y_i\in\{0,1\}$ the labels associated to the features, we desire our circuit to be trained so that, when prompted with $\xv_{i}$, it will be able to predict the associated label $y_i$. A suitable loss to achieve this goal is \textit{mean squared error}
\be
\label{eq:classification_loss}
\mathcal{L}(\thv,\vec{\alpha}) := \frac{1}{N_{\text{tr}}}\sum_{i}(f(\thv,\vec{\alpha},\xv_i)-y_i)^2 \,,
\ee
where $\thv$, $\alv$ are trainable parameters within the circuit and observable function. In particular, $f$ is assumed to be the normalized expectation value of a trainable observable function 
measured at the end of the PQC. 

For our classification task we used the circle dataset from the scikit-learn machine learning library~\cite{scikit_dataset}, where the data consists of $(x_1,x_2)$ coordinate pairs that are either inside or outside a circle.
Our model consists of a linear optical PQC and an observable function, both of which contain trainable parameters ($\thv$ and $\alv$ respectively). The PQC consists of 5 modes and 12 beamsplitters. Four beamsplitter encode the data, whereas the remaining 8 beamsplitter angles, denoted by $\thv$, are trainable. The circuit is always initialised with input state $\ket{\nv}=\ket{1,0,1,0,1}$. We display a schematic in Fig.~\ref{fig:QML_classification}.

We use a nonlinear observable with trainable parameters as a representative example of an observable that may be of interest for hybrid QML tasks. Our observable function is given by a trainable neural network that maps the detector outputs to a single value. The trainable observable used consists of a 5-to-4 dense linear layer, an \textit{ELU} activation function, a 4-to-1 dense linear layer and then another \textit{ELU} activation function. In total we thus have 29 trainable classical parameters (denoted by $\alv$). The output of this network is then rescaled to be within the range $[0,1]$ by a \textit{sigmoid} function.

To train our circuit parameters, we need to calculate the gradient of the loss function with respect to the $\theta_k$ parameters. Hence, we need to compute gradients of the function $f$ with respect to beamsplitter parameters. To achieve this, we use either the GPSR or FD, with the same total number of samples $\Ntot$ (600 per gradient calculation). In the case of the GPSR, the samples are allocated according to the 1-norm of the dirichlet kernels. For a single parameter update step, \textit{stochastic gradient descent} is used to optimize the $\thv$ angles with the $\alv$ parameters held constant, and then \textit{Adam} is used to optimize the classical parameters $\alv$, with $\thv$ held constant.

\begin{table}[]
    \centering
    \begin{tabular}{|c|c|c|c|c|}
    \hline
        Method & GPSR & $h=0.01$ & $h=0.1$ & $h=1$ \\
        \hline
        Accuracy [\%] & $90.25\pm0.21$ & $83.35\pm0.93$
& $90.14\pm0.25$ & $89.17\pm0.31$   \\
\hline
\end{tabular}
    \caption{\textbf{Binary classification accuracy.} We tabulate the average prediction accuracy for the 30 models trained using each method, and the standard error. We see that the GPSR and optimal FD trained models achieve the best classification accuracies.} 
    \label{tab:QML_accuracy}
\end{table}

In Fig.~\ref{fig:QML_classification}, we plot the loss function during training, averaged over thirty different random initializations. We observe that the lowest losses throughout training were achieved by models using the GPSR and FD with the heuristically optimal shift of $0.1$. The FD trained models with non-optimal shift size had higher variance in the loss curves and converged to higher loss values. All models were trained to achieve high classification accuracies (see: Tab.~\ref{tab:QML_accuracy}), and we see that both the GPSR and the optimal FD models achieved the highest accuracies, with no statistically significant differences between the two. The model accuracies achieved may be further improved with further refinement of the hyperparameters used, however these results are already sufficient to highlight a practical difference between the GPSR and FD methods with different shift sizes.
As discussed in Appendix~\ref{appendix:statistical_estimation} and observed in Fig.~\ref{fig:QML_classification}, step sizes for FD methods need to be fine-tuned depending on the task considered to achieve accurate gradient calculation; the GPSR will always be optimal. As such, the GPSR stands out from FD methods in that it can be applied to achieve good training results without prior knowledge about the problem.

\begin{figure}
    \centering
    \includegraphics[width=\linewidth]{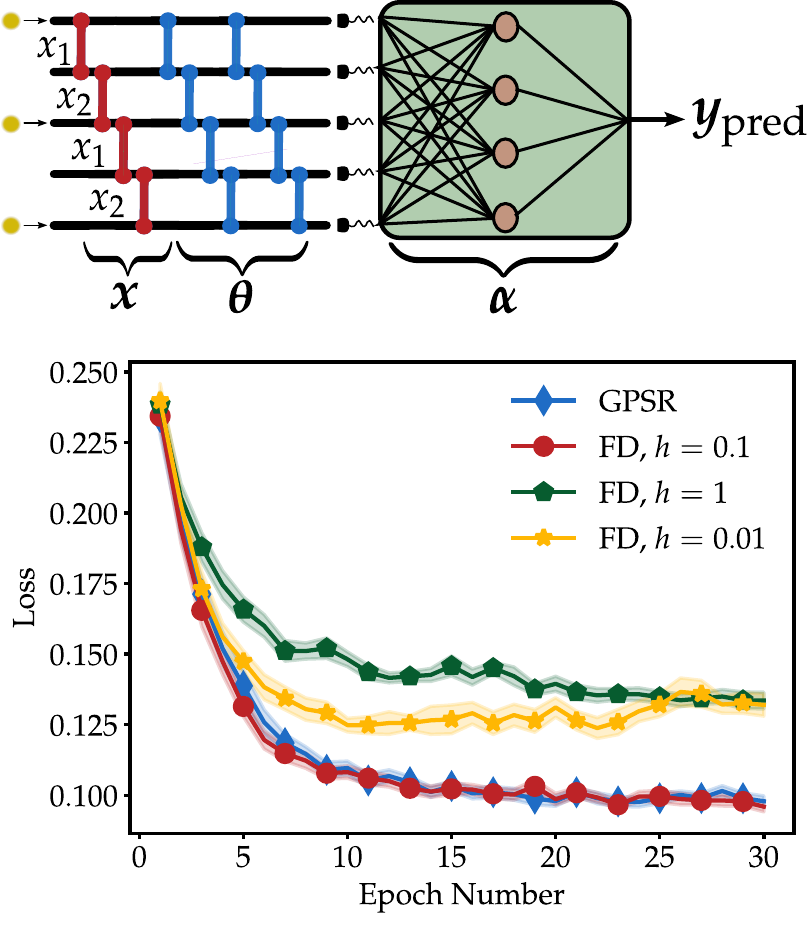}
    \caption{\textbf{Binary classfication.}
    \textit{Top:} A schematic of the variational quantum circuit used for the QML task considered here. The red beamsplitters are used to encode the datapoint features with data re-encoded once, and the blue beamsplitters have trainable parameters. The observable function (green box) consists of a trainable classical neural network. \textit{Bottom:}
    We plot the loss function defined in Eq.~\eqref{eq:classification_loss}, for the circle dataset, and for two different gradient calculation methods. In particular, we show: GPSR and FD with three different step sizes $h=0.01,0.1,1$. Based on the scaling of $h$ derived in Appendix~\ref{appendix:statistical_estimation} with $\Ntot=600$, $h=0.1$ is close to the optimal value and hence performs similarly to GPSR. However, choosing the wrong step size inhibits performance and leads to worse accuracy.}
    \label{fig:QML_classification}
\end{figure}

\subsection{Bell state generation circuit}
As another example, we study the impact of circuit parameters on a Bell state generation protocol. Qubits are commonly encoded in linear optical circuits via the \textit{dual-rail} encoding, where the states $\{\ket{0}_L, \ket{1}_L\}$ of a qubit are built from two modes and one photon via
\be
\ket{0}_L:=\ket{10}\,,\qquad \ket{1}_L:=\ket{01} \,.
\ee
Bell states are prototypical two-qubit maximally entangled states
\be
\begin{split}
\ket{\Phi^\pm} &:= \dfrac{1}{\sqrt{2}}(\ket{00}_L \pm \ket{11}_L) \,,\\
\ket{\Psi^\pm} &:= \dfrac{1}{\sqrt{2}}(\ket{01}_L \pm \ket{10}_L)\,,
\end{split}
\ee
also forming an orthonormal basis of a two-qubit Hilbert space.
However, entanglement among qubits is difficult to achieve in linear optical circuits with native operations only. Instead, many proposals rely on \textit{heralding} to artificially implement entangling, two-qubit interactions, hence bypassing the need for photon-photon interactions. On top of this, resource states such as Bell and GHZ states are central to some of the schemes proposed to realize universal LOQC~\cite{raussendorf2001oneway,raussendorf2003measurementbased,bombin2021interleaving}. 

In general, simulators provide a theoretical understanding of circuits that allow the construction of such states. However, experimentally implementing such architectures may be challenging and different sources of errors may lead to undesired behaviour. Therefore, it is also crucial to analyze how resource state generation circuits perform in practice. This can be achieved with the GPSR, which on top of being analytically exact, can also be run experimentally on real hardware.

With this in mind, we propose to study the sensitivity of a Bell state generation circuit generating the state $\ket{\Phi^+}$.
It has been shown that heralded Bell state generation requires at least four photons, where two are measured in auxiliary modes~\cite{stanisic2017generating}, so in this setting the GPSR method is crucial given the higher number of frequencies in the function $f(\cdot)$. In particular, we focus on the generation of Bell states with a six mode, four photon circuit originally proposed in Ref.~\cite{carolan2015universal} and depicted in Fig.~\ref{fig:bellstate_fidelity}. 

In this context, we define $\UC(\thv) \in \Ubb(d)$ as the circuit used for Bell state generation, and the output state $\ket{\psi(\thv)}=\UC(\thv)\ket{\nv}$ where ($\ket{\nv}=\ket{011110}$). The successful generation of the Bell state relies on measuring one photon in the third mode and one in the fourth. If this event happens, the state is then renormalized after applying the projection $\widehat{\Pi} = \Ibb_{1,2}\otimes\ket{11}\bra{11}_{3,4}\otimes \Ibb_{5,6}$
\be
    \ket{\Phi(\thv)} := \dfrac{\widehat{\Pi}\ket{\psi(\thv)}}{\norm{\widehat{\Pi}\ket{\psi(\thv)}}}
\ee
Finally, we define the observable to be the overlap with the ideal state, i.e. $O=\ket{\Phi^+}\bra{\Phi^+}$ so that
\be
 \expval{O} = \abs{\bra{\Phi^+}\ket{\Phi(\thv)}}^2 := \FC(\Phi^+) \,.
\ee
We can then take the derivative with respect to a given parameter $\theta_k$
\be
\begin{split}
\partial_k \FC(\Phi^+) &= \dfrac{\partial_k\langle(\ket{\Phi^+}\bra{\Phi^+}\otimes \widehat{\Pi})\rangle_{\thv}}{\langle\widehat{\Pi}\rangle_{\thv}} \\ 
&-\dfrac{\langle(\ket{\Phi^+}\bra{\Phi^+}\otimes \widehat{\Pi})\rangle_{\thv}}{\langle{\widehat{\Pi}\rangle_{\thv}^2}}\partial_k\langle\widehat{\Pi}\rangle_{\thv}
\end{split}
\ee
where $\expval{\cdot}_{\thv}:=\bra{\psi(\thv)}\cdot\ket{\psi(\thv)}$. In Fig.~\ref{fig:bellstate_fidelity}, we plot the different derivatives of the fidelity with respect to the parameters in the circuit, as functions of the distance away from the parameter's optimal value. Near the optimal value, we highlight that derivatives for $\theta_4,\theta_5$ are larger than the other ones, implying that such parameters have a higher impact on the overall fidelity of the state.
\begin{figure}
    \centering
    \includegraphics[width=\linewidth]{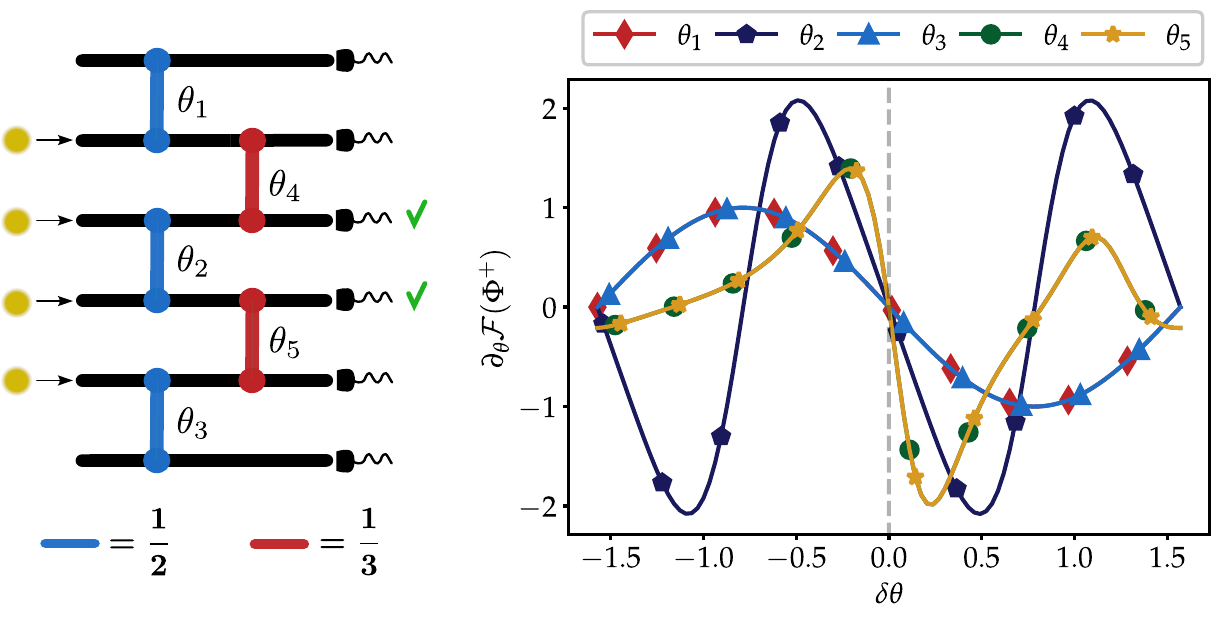}
    \caption{\textbf{Bell state generation derivatives of the fidelity.} \textit{Left:} We show the Bell state generation circuit proposed in Ref.~\cite{carolan2015universal}, where four photons are in the four central modes, interfere through five beamsplitters with optimzal values at $\cos^2{\theta_{\text{opt}}}=1/2$ (blue) and $\cos^2{\theta_{\text{opt}}}=1/3$ (red). Upon detecting two photons in the two central modes, a Bell state is created. \textit{Right:} We plot the derivative of the fidelity of the Bell state generation with respect to the parameters in the circuit, as a function of the distance away from the optimal values $\delta\theta = \theta-\theta_{\text{opt}}$.} 
    \label{fig:bellstate_fidelity}
\end{figure}

\section{Discussion}\label{sec:discussion}
Here, we discuss some of the advantages and limitations of the results presented in this work. The GPSR protocol provides an exact expression for the gradient of linear optical circuits that can be experimentally employed to reconstruct derivatives, without the need to rely on finite difference approximations. However, though we have proposed methods for reducing the number of circuit evaluations required, the linear scaling of this approach with the number of photons can still be a practical limitation for large circuits. We anticipate that this method will thus be most useful in cases where high-quality gradients are particularly important, such as in the final stages of an optimization task or for resource state generation.

We also note that, for optimization tasks, while employing the GPSR may require more samples to yield an accurate gradient, there might be a benefit in convergence speed. However, we leave as future work to explore this trade-off in more detail, as well as strategies to try and  avoid poor local minima~\cite{anschuetz2022beyond, bittel2021training} and barren plateaus~\cite{mcclean2018barren, larocca2024review} which were not discussed here. 
Possible future works may investigate how to combine the GPSR for linear optics with stochastic approaches~\cite{banchi2020training,sweke2020stochastic,kubler2020adaptive}. A more detailed analysis of the approximation capabilities of Dirichlet kernels could also bring insight on when is it appropriate to truncate the number of Fourier frequencies, and hence parameter shifts.

An important experimental consideration is the presence of imperfections such as photon loss or partial distinguishability. We anticipate that the GPSR is still valid for most common sources of imperfections.
For instance, in Appendix~\ref{appendix:lossy_GPSR} we show that the GPSR protocol holds also for interferometers with photon loss.
Similarly, linear optics with partially distinguishable photons~\cite{tichy2015sampling} was demonstrated to be equivalent to boson sampling with fewer ideal photons~\cite{renema2018efficient}, so the GPSR will hold as well.

\textit{Note added:} During the preparation of this manuscript, we became aware of recent work in Ref.~\cite{cimini2024variational} in which the authors present a similar result for a restricted version of the GPSR for circuits involving two photons, with a supplementary note deriving a generalization similar to ours.  Our work goes further by giving an in-depth analysis of the main result, including sample complexity scaling estimations and protocols to reduce the number of parameter shifts. We also acknowledge concurrent work in Refs.~\cite{pappalardo2024photonic,hoch2024variational}.

\section{Conclusion}
In this work, we showed how one can reconstruct partial derivatives of functions in linear optical circuits. To do so, we first derived how an arbitrary expectation value can be expanded into a finite Fourier series with respect to a given parameter in the circuit. Having done so, we then reconstructed the expectation value, and its derivatives, as a univariate function by evaluating the function at shifted values of the parameter. In order to reduce the required number of expectation values, we proposed two strategies that allow us to reduce the number of parameter shifts, the first one based on a light cone argument which largely depends on the specifics of the circuit, and the second focusing on the form of the observable. Finally, we showcased the use the GPSR for optimizing parameters in a PQC and to measure the impact of component imperfections in resource state generation.

More generally, we hope that this work might open up further exploration on applications with linear optics and single photons in the context of VQC. These include problems in optimization and control, chemistry, ML, unitary learning, noise characterization and quantum compilation.

\section*{Acknowledgements}
We thank Alex Martin and the rest of the team at ORCA Computing for insightful discussions and comments. Z.H. acknowledges support from the Sandoz Family Foundation-Monique de Meuron program for Academic Promotion.
This work was supported by the Quantum Data
Centre of the Future, project number 10004793, funded by the Innovate UK
Industrial Strategy Challenge Fund (ISCF).

\bibliography{biblio.bib, quantum.bib}
\newpage
\onecolumngrid
\setcounter{theorem}{0}
\setcounter{lemma}{0}
\setcounter{corollary}{0}

\appendix
\vspace{0.5in}
\begin{center}
	{\Large \bf Appendix} 
\end{center}

\section{Trigonometric Interpolation}\label{appendix:trig_interpolation}
In this Appendix we derive some key results which are useful in order to derive the main result of the GPSR rule, i.e. Theorem~\ref{thm:GPSR}, as well as other secondary results in this work. A comprehensive introduction on trigonometric interpolation can be found in Refs.~\cite{zygmund2003trigonometric,atkinson1991introduction}.

Generally, we consider a function of the form
\be\label{eq:appendix_fourier_func}
f(\theta) = \sum_{k=-R}^{R} c_k e^{ik\theta} = a_0 + \sum_{k=1}^R a_k \cos(k\theta) + b_k \sin(k\theta) \,,
\ee
where $a_{k} = (c_{k}+c_{-k})$, $b_{k} = i(c_{k}-c_{-k})$. As shown in the main text, this resembles the expectation value of an arbitrary observable in linear optics, as a function of a given parameter $\theta$.
For completeness, we give here a proof for the result of Lemma~\ref{lemma:trig_interp}. Let us recall the claim
\begin{lemma}[Trigonometric Interpolation, Ref.~\cite{atkinson1991introduction}]\label{appendix_lemma:trig_interp}
Suppose a function of the form as in Eq.~\eqref{eq:appendix_fourier_func}. Suppose furthermore that we are given the values of $f$ at $2R+1$ equidistant points in a $2\pi$ interval, e.g. we have knowledge of $\{(\theta_\mu,f(\theta_\mu))\, |\, \theta_\mu=2\pi \mu/(2R+1),\, \mu=-R,\dots,R\}$. Then one can show that
\be
    c_k = \dfrac{1}{2R+1}\sum_{\mu=-R}^{R} e^{-ik\theta_\mu}f(\theta_\mu)\hspace{2pt} \forall\, \mu=-R,\dots,R
\ee
\end{lemma}
\begin{proof}
First of all, note that 
\be
    \sum_{\mu=-R}^R e^{i\theta_\mu} = \sum_{\mu=0}^{2R} e^{i\theta_\mu} = \sum_{\mu=0}^{2R} z^{\mu} \,,
\ee
with $z = e^{i2\pi/(2R+1)}$. For $z\neq 1$, by the geometric series identity this is equal to
\be
\sum_{\mu=0}^{2R} z^{\mu} = \dfrac{z^{2R+1}-1}{z-1} = \dfrac{e^{i2\pi}-1}{e^{i2\pi/(2R+1)} -1} = 0 \,.
\ee
For $z=1$, we easily verify that instead the sum amounts to $2R+1$. Let us now look at the set of interpolating points equations
\be
    \sum_{k=-R}^R c_k e^{ik\theta_\mu} = f(\theta_\mu)\,,\qquad \mu=-R,\dots,R \,.
\ee
Multiplying each term in the sum by $e^{-ij\theta_\mu}$ and summing over $\mu=-R,\dots,R$
\be\label{eq:trig_interpolation_proof}
    \sum_{\mu=-R}^R \sum_{k=-R}^R c_k e^{i(k-j)\theta_\mu} = \sum_{\mu=-R}^R e^{-ij\theta_\mu}f(\theta_\mu)\,.
\ee
Reversing the order of summation on the LHS of the equation,  we get
\be
    \sum_{k=-R}^R c_k \sum_{\mu=-R}^R e^{i(k-j)\theta_\mu} = \sum_{k=-R}^R c_k \delta_{kj}(2R+1) = c_j (2R+1) \,.
\ee
Inserting this result in Eq.~\eqref{eq:trig_interpolation_proof}, we finally complete the proof
\be
    c_j = \dfrac{1}{2R+1}\sum_{\mu=1}e^{-ij\theta_\mu}f(\theta_\mu) \,.
\ee
\end{proof}

Consider the \textit{Dirichlet kernel}, defined as
\begin{align}
    D(x) &=\dfrac{1}{2R+1}\sum_{k=-R}^R e^{ikx}
    = \dfrac{1}{2R+1}\bigg(1+\sum_{k=1}^R e^{ikx}+e^{-ikx}\bigg) \\
    &= \dfrac{1}{2R+1}\bigg(1+2\sum_{k=1}^R \cos{(kx)}\bigg)\,, \qquad D(n\pi) = 1 \,\forall n\in \Nbb \,,
\end{align}
which can also be brought to the form in Eq.~\eqref{eq:def_dirichlet_kernel}. Then, inserting the result of Lemma~\ref{appendix_lemma:trig_interp} into the definition of $f(\theta)$ in Eq.~\eqref{eq:appendix_fourier_func}, then one can show that the function is expressed as a sum of Dirichlet kernels
\be
    f(\theta)= \sum_{\mu=-R}^R f(\theta_\mu)D(\theta-\theta_\mu) 
\ee
From here, the main result of this work is derived.

We also evaluate the derivatives and their values at $x=n\pi$, $n\in\Nbb$, as it will be useful for Appendix~\ref{appendix:statistical_estimation}. In fact
\be
D'(x) = -\dfrac{2}{2R+1}\sum_{k=1}^R k\sin{(kx)}\,, \qquad D'(n\pi) = 0 \,.
\ee
The second derivative is given by
\be
D''(x) = -\dfrac{2}{2R+1}\sum_{k=1}^R k^2\cos{(kx)}\,.
\ee
Let us evaluate the second derivative at the points $x=2n\pi$, $n\in \Nbb$
\be\label{eq:D''_at_2npi}
D''(2n\pi) = -\dfrac{2}{2R+1}\sum_{k=1}^R k^2 = -\dfrac{1}{3}R(R+1)
\ee
and for $x=(2n+1)\pi$, $n\in \Nbb$ (for completeness)
\be
D''((2n+1)\pi) = -\dfrac{2}{2R+1}\sum_{k=1}^R (-1)^k k^2 = \dfrac{(-1)^{R+1} 2R(R+1)}{2R+1}
\ee

\subsection{First order derivative reconstruction} Here we loosely follow the derivation of Ref.~\cite{zygmund2003trigonometric}. An alternative derivation is provided in Ref.~\cite{wierichs2022general} where they first focus on reconstructing odd functions, and then derive the result for the first-order derivative of an arbitrary function defined from a qubit-based quantum circuit. Similarly to the general case, we consider the task of reconstructing a function of the form
\be
f(\theta) = \dfrac{1}{2}a_0 + \sum_{k=1}^R a_k \cos(k\theta) + b_k \sin(k\theta) \,.
\ee
Given $2R$ points, the system will be under-determined since for the full reconstruction we would need $2R+1$ points to resolve for the coefficients $\{a_0,a_1,b_1,\dots,a_R,b_R\}$. However, $2R$ are enough to determine the first-order derivative.

Consider the \textit{modified Dirichlet} kernel
\be
    D^*(\theta) = \dfrac{1}{2R}+\dfrac{1}{2R}\cos(R\theta)+\sum_{\mu=1}^{R-1}\cos(k\theta) = \dfrac{\sin(R\theta)}{2R\tan(\theta/2)}
\ee
Consider the points $\{\theta_\mu=\dfrac{(2\mu-1)\pi}{2R}\,,\,\mu=1,\dots,2R\}$. Then, the modified Dirichlet kernel satisfies $D^*(\theta_{\mu}-\theta_{\mu'}) = \delta_{\mu\mu'}$. It can then be shown that
\be
    f(\theta) = a_n\cos(R\theta) + \sum_{\mu=1}^{2R}f(\theta_\mu) D^*(\theta-\theta_\mu)
\ee
We then find that
\be
f'(0) = \sum_{\mu=1}^{2R} f(\theta_\mu)(D^*)'(-\theta_\mu) = \sum_{\mu=1}^{2R} f(\theta_\mu)\dfrac{(-1)^{\mu+1}}{4R\sin^2(\theta_\mu /2)}
\ee
From which we can deduce
\be
    f'(\theta) = \sum_{\mu=1}^{2R}f(\theta+\theta_\mu)\dfrac{(-1)^{\mu+1}}{4R\sin^2(\theta_\mu /2)}
\ee
While this has a similar form to the main result with the standard Dirichlet kernels, it turns out that this formulation achieves a better scaling in terms of number of samples. To see this in more detail, the next section will be concerned with the shot budget required to reconstruct the first-order derivative within a desired additive error.

\subsection{Presence of noise in the system}\label{appendix:lossy_GPSR}
In this section, we discuss the presence of noise in the system. Experimental realizations of linear optical circuits are in fact faced with photon loss in the system. For a system with $n$ photons, the architecture has some probability of losing a certain amount of photons, a probability which generally depends on the specifics of the platform. For uniform loss with transmission coefficient $\eta$, loss commutes with the unitary $\UC(\thv)$ and therefore can be applied directly to the initial state $\rho_{\nv}:=\ket{\nv}\bra{\nv}$~\cite{garciapatron2019simulating}. We define the transformation where a photon in a given mode escapes and is injected into the environment
\be
a\ad \mapsto \sqrt{\eta}a\ad +\sqrt{1-\eta}e\ad
\ee
where $e\ad$ is the creation operator of the environment. The action on an $n$-photon state is given by
\be
\begin{split}
\dfrac{1}{\sqrt{n!}}(a\ad)^n\ket{0}_a\otimes\ket{0}_e &\mapsto \dfrac{1}{\sqrt{n!}}\bigg(\sqrt{\eta}a\ad + \sqrt{1-\eta}e\ad\bigg)^n\ket{0,0} \\
&=\dfrac{1}{\sqrt{n!}}\sum_{k=0}^n\binom{n}{k}\eta^{k/2}(1-\eta)^{\frac{n-k}{2}}\sqrt{k!}\sqrt{(n-k)!}\ket{k,n-k} \\
&=\sum_{k=0}^n\binom{n}{k}^{1/2}\eta^{k/2}(1-\eta)^{\frac{n-k}{2}}\ket{k,n-k} \,.
\end{split}
\ee
Tracing out the environment, we get that the pure loss channel acts on an $n$-photon state $\rho_n=\ket{n}\bra{n}$ as
\be
\Lambda_{\eta}(\rho_n) := \sum_{k=0}^n\binom{n}{k}\eta^k(1-\eta)^{n-k} \rho_{k} = \sum_{k=0}^n p^{\eta}_{n,k} \rho_{k} \,,
\ee
where we deonte by $p^{\eta}_{n,k}$ the probability of being in state $\rho_k:=\ket{k}\bra{k}$.
Now, the same channel acting on $m$ modes gives
\be
\Lambda_{\eta}^{\otimes m}(\rho_{\nv})= \bigotimes_{i=1}^m \bigg(\sum_{k_i=0}^{n_i} p^{\eta}_{n_i,k_i} \rho_{k_i} \bigg) = \sum_{k_1=0}^{n_1}\dots\sum_{k_m=0}^{n_m}\bigg(\prod_{i=1}^m p^{\eta}_{n_i,k_i}\bigg)\ket{\kv}\bra{\kv} =
\sum_{\vec{k}} p^{\eta}_{\nv,\kv}\ket{\kv}\bra{\kv} \,,
\ee
where $p^{\eta}_{\nv,\kv}$ denotes the product in the parentheses, and the final sum runs over all $\ket{\kv}$ Fock states such that $k_i \leq n_i$.
The expectation value of an observable $O$ measured with a parametrized, noisy device will be
\be
\begin{split}
\expval{O}_{\eta} &= \Tr[\UC(\thv)\Lambda_{\eta}^{\otimes m}(\rho_{\nv})\UC\ad(\thv)O] \\
& =  \sum_{\kv} p^{\eta}_{\nv,\kv} \Tr[\UC(\theta)\rho_{\kv}\UC\ad(\theta)O] \,,
\end{split}
\ee
where $\expval{\cdot}_{\eta}$ indicates the expectation value under noise $\eta$. Importantly, each $\Tr[\cdot]$-term in the sum is the expectation value of the same observable but with a $k$-photon input state. By applying similar arguments to the main derivation, it is then possible to show that each single $\Tr[\cdot]$-term can be expanded as
\be
    \Tr[\UC(\theta)\rho_{\kv}\UC\ad(\theta)O] = \sum_{\omega=-k}^k c_{\omega,k}e^{i\omega\theta}=
    \sum_{\omega=-n}^n c_{\omega,k}e^{i\omega\theta} \,,
\ee
where the $c_{\omega,k}$'s terms, similarly to the lossless case, depend on the unitary and the observable. In the second equality we extended the sum by simply considering $c_{\omega,k}\equiv 0$ for $|\omega| > k$. The expectation value then takes the form
\be
\begin{split}
    \expval{O}_{\eta} = \sum_{\omega=-n}^n\bigg(\sum_{k=0}^n c_{\omega,k}\binom{n}{k}\eta^k(1-\eta)^{n-k} \bigg) e^{i\omega\theta}
    =\sum_{\omega=-n}^n c_{\omega,\eta} e^{i\omega\theta} \,.
\end{split}    
\ee
This form of the expectation value shows that all the results derived for the ideal case also hold in the case of circuits with uniform loss.

In the case of mode-dependent losses, the dynamics becomes more complex as the losses cannot commute and be shifted to the beginning (or end) of the circuit. However, it is still reasonable to assume that, in practice, we end up with the same Fourier spectrum. In fact, very generally, we can consider the unitary dilation of a lossy linear optics circuit, where the final state of the $m$ modes is given by
\be\label{eq:def_loss}
    \rho_{\text{out}} = \Tr_{E}[\widetilde{\UC}(\theta,\vec{\eta})\ad(\rho_{\nv}\otimes \ket{\vec{0}}\bra{\vec{0}}_{E})\widetilde{\UC}(\theta,\vec{\eta})]
\ee
where $E$ denotes the entire environment system, and $\widetilde{\UC}(\theta,\vec{\eta})$ is a unitary specified from the ideal circuit, as well as additional beamsplitters coupling the $m$ modes and the environment (which act in the same way as Eq.~\eqref{eq:def_loss}, with distinct parameters $\eta_i$). We denote $\vec{\eta}=(\eta_1,...\eta_{m})$ the loss parameters. An observable measured on the $m$ modes yields the following expectation value
\be
\expval{O} = \Tr_{m}[\rho_{\text{out}}O] = \Tr[\widetilde{\UC}(\theta,\vec{\eta})\ad(\rho_{\nv}\otimes \ket{\vec{0}}\bra{\vec{0}}_{E})\widetilde{\UC}(\theta,\vec{\eta}) (O\otimes\Ibb_{E})] 
\ee
While needing to consider an enlarged space containing environment modes, and more complex dynamics which include the loss events, the state $\rho_{\nv}\otimes \ket{\vec{0}}\bra{\vec{0}}_{E}$ is still a $n$-photon state and the unitary can still be cast so that the parametrized operation can be singled out and expanded to achieve a spectrum $\{-n,\dots,n\}$.

\section{Statistical estimation and accuracy}\label{appendix:statistical_estimation}
We will now study the accuracy of different gradient estimators that can be used in parametrized quantum circuits in linear optical systems. As a first step, an important assumption we make is that the physical variance $\sigma^2$ does not depend on the parameters $\thv$. Alternatively, we can think of $\sigma^2$ as the variance of $O$ maximised over the parameter space (and what follows can be thought of as a worst-case scenario).

Generally, an estimator $\widehat{\EC}_{O}$ of the expectation value of an observable $\expval{O}$ can be biased such that
\be
\widehat{\EC}_{O} = \expval{O} +\widehat{b} \,,\qquad \Ebb[\widehat{\EC}_{O}] = \expval{O} +b \,.
\ee
As an example, the GPSR in the infinte sampling regime is an exact result and hence its estimator will be unbiased. A finite difference approach, instead, will generally not be a unbiased estimator and will have a non-zero $b$ depending on the resolution size. We will show this in more detail in section~\ref{appendix:finite_difference}. 
One of the key figures of merit one can consider when estimating the distance away of an estimator from the desired value, is the mean-squared error (MSE)
\be
\begin{split}
\mse{\widehat{\EC}_O} &= \Ebb[(\widehat{\EC}_O-\expval{O})^2] = \Ebb[\widehat{\EC}^2_O]-2\Ebb[\widehat{\EC}_O]\expval{O} + \expval{O}^2 \\
&= \Var[\widehat{\EC}_O] +b^2 \,,
\end{split}
\ee
where in the last equality we added and subtracted the squared of the expectation value of $\widehat{\EC}_O$ to get the result. Hence, when considering the accuracy of estimators, there are in general two sources of errors: the bias $b$ represents a constant error that persists even when the number of samples $\Ntot\rightarrow\infty$, while the variance is caused by the statistical fluctuations intrinsic to the measurement process of quantum system.

\subsection{Generalized parameter-shift rule}\label{appendix:gpsr}
Let us derive how many samples per expectation value $N_s$, we need to get an estimate within distance $\epsilon$ from the actual value of the derivative.
First of all, we compute the physical variance of the trigonometric-interpolated derivative\footnote{To differentiate GPSR and finite difference, will use subscript to specify the strategy used.}
\be
    \Var[f'_{\text{GPSR}}] =\Var\bigg[\sum_{\mu=1}^R(f(\theta+\theta_\mu)-f(\theta-\theta_\mu))\dfrac{(-1)^{\mu+1}}{2\sin{(\theta_\mu / 2)}}\bigg] 
    =\sigma^2 \sum_{\mu=1}^R \dfrac{1}{2\sin^2{(\theta_\mu/2)}} \,,
\ee
where we assumed that $f(\theta + \theta_\mu)$ and $f(\theta - \theta_\mu)$ are uncorrelated. Note that the derivative does depend on $\theta$, we dropped the explicit dependence for notational convenience. Hence, in order to make sense of this quantity, we need to compute the sum over $\mu$. To do this, note that using the main GPSR result (Eq.~\eqref{eq:gpsr}) for $f(\theta)=D'(\theta)$ evaluated at $\theta=0$ yields
\be
D''(0) = \sum_{\mu=1}^R (D'(\theta_\mu) -D'(-\theta_\mu))\dfrac{(-1)^{\mu+1}}{2\sin{(\theta_\mu / 2)}} = \sum_{\mu=1}^R 2D'(\theta_\mu)\dfrac{(-1)^{\mu+1}}{2\sin{(\theta_\mu / 2)}} = - \sum_{\mu=1}^R \dfrac{1}{2\sin^2{(\theta_\mu/2)}}
\ee
Using the result of Eq.~\eqref{eq:D''_at_2npi}, we see that
\be
    \Var[f'_{\text{GPSR}}] = \dfrac{\sigma^2 R(R+1)}{3}
\ee
We are now in the position to determine the number of samples required to estimate the gradient within accuracy $\epsilon$. In order to estimate the function at the shifted values $f(\theta\pm\theta_{\mu})$ given $N_s$ measurements, we use the unbiased estimator
\be\label{eq:shifted_function_estimator}
\widehat{f}_{\theta\pm\theta_\mu} = \dfrac{1}{N_s}\sum_{s=1}^{N_S} \lambda_{\nv(s)} \,, \qquad \Ebb[\widehat{f}_{\theta\pm\theta_\mu}] =f(\theta\pm\theta_\mu) \,, \,\Var[\widehat{f}_{\theta\pm\theta_\mu}] = \dfrac{\sigma^2}{N_s}
\ee
where $\lambda_{\nv(s)}$ is the eigenvalue of $O$ associated to outcome state $\nv(s)$ measured at the $s$-th measurement. From this, we can then define the estimator for the gradient
\be
    \widehat{f'}_{\text{GPSR}} = \sum_{\mu=1}^R (\widehat{f}_{\theta+\theta_\mu} - \widehat{f}_{\theta-\theta_\mu})\dfrac{(-1)^{\mu+1}}{2\sin{(\theta_\mu/2)}} \,, \qquad \Ebb[\widehat{f'}_{\text{GPSR}}] =f'(\theta) \,, \,\Var[\widehat{f'}_{\text{GPSR}}] = \dfrac{\sigma^2R(R+1)}{3N_s} \,,
\ee
where to compute the variance, we again assumed that the estimators $\widehat{f}_{\theta + \theta_\mu}$ and $\widehat{f}_{\theta - \theta_\mu}$ contain samples which are uncorrelated. Now, note that for such an unbiased estimator the MSE is given by
\be
\mse{\widehat{f'}_{\text{GPSR}}} = \dfrac{\sigma^2R(R+1)}{3N_s}\,.
\ee
Then, if we desire the MSE to be within a certain accuracy $\epsilon^2$, we will require
\be
    \mse{\widehat{f'}_{\text{GPSR}}} \leq \epsilon^2 \implies N_s \in \OC\left(\dfrac{\sigma^2 R(R+1)}{3 \epsilon^2} \right)
\ee
yielding a quadratic scaling with respect to the number of photons in the general case of the GPSR. In terms of the total number of calls to any given circuit, i.e. $\Ntot=2RN_s$, we instead have a cubic scaling
\be
\Ntot \in \OC\bigg(\dfrac{2\sigma^2 R^2(R+1)}{3 \epsilon^2} \bigg)
\ee
In general, this scaling of the sample complexity is not optimal for the goal of computing derivatives, and more efficient strategies could be employed. 

In particular, one could consider reconstruction of the derivative by distributing the total number of samples $N_\text{tot}$ according to the weight of each term in the sum of the GPSR. Then, the number of samples for each estimator $\hat{f}_{\theta\pm\theta_{\mu}}$ defined in Eq.~\eqref{eq:shifted_function_estimator} will be
\be
N_{s,\mu} = \frac{N_\text{tot} |y_\mu|}{2\norm{\vec{y}}_1}\,,\qquad \vec{y} = \dfrac{1}{2}\bigg(\dfrac{1}{\sin(\theta_1 /2)},\dots, \dfrac{(-1)^{n+1}}{\sin(\theta_n /2)}\bigg) \,,
\ee
where $\vec{y}$ is the vector containing the weights given by the derivatives of the Dirichlet kernel and $\norm{\cdot}_1$ is the 1-norm. For a vector $\vec{v}=(v_1,\dots,v_R)$, it is defined as $\norm{\vec{v}}_1 = \sum_{\mu=1}^R |v_\mu|$
In this case, the variance of the estimator $\widehat{f'}_{\text{GPSR}}$ is
\be
\begin{split}
    \Var[\widehat{f'}_{\text{GPSR}}] &= 
    \sum_{\mu=1}^R (\Var[\widehat{f}_{\theta+\theta_\mu}] + \Var[\widehat{f}_{\theta-\theta_\mu}])y_\mu^2 \\
    &= \sum_{\mu=1}^R \dfrac{4\sigma^2\norm{\vec{y}}_1}{\Ntot|y_\mu|}y_\mu^2 = \dfrac{4\sigma^2\norm{\vec{y}}_1}{\Ntot}\sum_{\mu=1}^R|y_\mu| \\
    &= \dfrac{4\sigma^2\norm{\vec{y}}^2_1}{\Ntot} \,,
    \end{split}
\ee
To evaluate the variance, we need to compute the 1-norm of th vector $\vec{y}$. This is very difficult, if not impossible, in practice. However, Ref.~\cite{blagouchine2024finite} showed a particularly tight approximation to this kind of sum. In particular, it can be shown that 
\be
\begin{split}
\norm{\vec{y}}_1 &= \sum_{\mu=1}^R \dfrac{1}{2\sin{\bigg(\dfrac{\pi\mu}{2R+1}\bigg)}} = \dfrac{1}{4}\sum_{\mu=1}^{2R} \dfrac{1}{\sin{\bigg(\dfrac{\pi\mu}{2R+1}\bigg)}} \\
&= \dfrac{1}{2}\dfrac{2R+1}{\pi}\bigg(\ln{\bigg(\dfrac{2(2R+1)}{\pi}\bigg)} + \gamma\bigg) + \OC(1) \\
&\sim \dfrac{1}{2} R\ln{(R)} + \OC(R) \,.
\end{split}
\ee
In this case, we are then able to achieve a more favourable scaling with respect to $R$
\be\label{eq:scaling_1norm_dirichlet}
\mse{\widehat{f'}_{\text{GPSR}}} \leq \epsilon^2 \Longrightarrow \Ntot \in \OC\bigg(\dfrac{\sigma^2 R^2 \ln^2(R)}{\epsilon^2}\bigg) \,.
\ee
\textbf{GPSR with modified Dirichlet kernels} Furthermore, one can consider reconstruction of the derivative via the modified Dirichlet kernels. A full derivation is provided in Appendix A.4 of Ref.~\cite{wierichs2022general} and we report here the scaling achieved when distributing the samples $\Ntot$ according to the $1$-norm 
\be
    \Ntot \in \OC\bigg(\dfrac{\sigma^2 R^2}{\epsilon^2}\bigg) \,,
\ee
hence obtaining a further $\ln^2(R)$ saving with respect to the the result in Eq.~\eqref{eq:scaling_1norm_dirichlet}.
\subsection{Finite difference}\label{appendix:finite_difference}
In this section, we follow the derivation of Ref.~\cite{mari2021estimating} to estimate the optimal shot budget for a finite difference (FD) approach. We recall the definition of the FD
\be
f'_{\text{FD}} = \dfrac{f(\theta+\delta)-f(\theta-\delta)}{2\delta} \simeq f'(\theta) \,,
\ee
where the approximate equality is true as long as $\delta$ is small enough, and we again dropped the $\theta$ dependence for notational convenience. Assuming again that $\sigma^2$ does not depend on $\theta$, we have $\Var[f'_{\text{FD}}]=\sigma^2/2\delta^2$. This implies that for an estimator allocating $N_s$ samples for each of the two terms in the numerator of the FD, we will have
\be
\Var[\widehat{f'}_{\text{FD}}] = \dfrac{\sigma^2}{2\delta^2 N_s} \,.
\ee
We now compute the bias
\be
\begin{split}
b &= \Ebb[\widehat{f'}_\text{FD}] -f'(\theta) = \dfrac{f(\theta+\delta)-f(\theta-\delta)}{2\delta} - f'(\theta) \\
&= \dfrac{f(\theta)+f'(\theta)\delta + \frac{1}{2!}f''(\theta)\delta^2 + \frac{1}{3!}f'''(\theta)\delta^3-[f(\theta)-f'(\theta)\delta + \frac{1}{2!}f''(\theta)\delta^2 - \frac{1}{3!}f'''(\theta)\delta^3]+ \OC(\delta^5)}{2\delta} -f'(\theta) \\
&=\dfrac{f'''(\theta)\delta^2}{3!} +\OC(\delta^4)\,,
\end{split}
\ee
where the Taylor expansion of the two terms is valid in the regime $\delta\ll 1$. The MSE is then
\be
\mse{\widehat{f'}_\text{FD}} = \dfrac{f'''(\theta)^2\delta^4}{(3!)^2} + \dfrac{\sigma^2}{2\delta^2 N_s}
\ee
The optimal step size $\delta$ minimizing the MSE is then
\be
\delta_{\text{opt}} = \bigg[\dfrac{9\sigma^2}{f'''(\theta)^2 N_s}\bigg]^{1/6}\,.
\ee
So the minimized MSE will be
\be
    \mse{\widehat{f'}_{FD}} = \dfrac{3}{4}\bigg(\dfrac{f'''(\theta)^2\sigma^4}{9 N_s^2}\bigg)^{1/3}
\ee
Again, for an accuracy $\epsilon^2$, the sample complexity will be
\be
N_s \in \OC\left(\dfrac{\sqrt{3}\abs{f'''(\theta)}\sigma^2}{8\epsilon^3}\right) \xRightarrow{N_\text{tot}=2N_s} N_\text{tot} \in \OC\left(\dfrac{\sqrt{3}\abs{f'''(\theta)}\sigma^2}{4\epsilon^3}\right) 
\ee

Given we need knowledge of the third order derivative, it is in general difficult to exactly determine the required number of samples. Moreover, the result holds as long as the Taylor expansion is valid, i.e. $\delta \ll 1$. However, we remark that while $\Ntot$ contains a cubic scaling with respect to  $\epsilon^{-3}$ scaling.

\section{Simplifications to the GPSR}
\subsection{Number of shifts depends on the first unitary}\label{appendix:gpsr_reduced1}

 Depending on the position of the parameter with respect to which we want to compute the gradient, it may be possible to reduce the number of shifts in the GPSR formula. We formalise this by considering that the unitary $\WC$ preceding the phase $\theta$ in mode $\nu$ might not act on the modes altogether, but there might exist subsets of the modes which do not interact. To be more clear, we suppose $\WC = \WC_A\otimes \WC_B$, where $A \cap B =\emptyset$. The modes and photons in region $A$ ($B$) are $m_A$ ($m_B$) and $n_A$ ($n_B$), respectively. Furthermore, we suppose that the mode $\nu$ onto which the parametrized phase is applied is contain in region $A$. A drawing of the setup is provided in Fig.~\ref{fig:resource_savings}. We then have the following result

\begin{corollary}[Light cone bound on the number of shifts]
    Given the same assumptions of Theorem~\ref{thm:GPSR}, and additionally supposing that the size of the light cone generated by $\WC\ad$ is $m_A$, such that $\WC=\WC_A\otimes \WC_B \in \Ubb(d)$, and we recall that $\PC_{\theta}$ is the unitary representing a phase applied to mode $\nu \subset A$, so it holds that $\PC_{\theta}\rightarrow \PC_{\theta} \otimes\Ibb_B$. We then have that the GPSR is reduced to
\be
f^{(k)}(\theta) = \sum_{\mu=-n_A}^{n_A}f(\theta + \theta_\mu)D^{(k)}(-\theta_\mu) \,,
\ee
where $\theta_\mu = \frac{2\pi\mu}{2n_A+1}$, and $n_A$ is the number of photons contained in region $A$ in the initial state. Similarly, the result of Corollary~\ref{corollary:GPSR_odd} reduces to
\be
f'(\theta) = \sum_{\mu=1}^{2n_A} f(\theta + \theta_\mu)\dfrac{(-1)^{\mu+1}}{4R\sin^2(\theta_\mu/2)} \,,
\ee
where $\theta_\mu = \frac{(2\mu-1)\pi}{2n_A}$.
\end{corollary}

\begin{proof}
    Let us note that for Fock states, the following states can be put in tensor product form $\ket{\nv}=\ket{\nv_A}\otimes\ket{\nv_B}$ and $|\sv'\rangle=|\sv'_A\rangle\otimes|\sv'_B\rangle$. Let us then compute again the amplitude appearing in Eq.~\eqref{eq:expval_thirdline} when we expand the expectation value of an arbitrary obsrvable
    \be\label{eq:reduced_amplitude}
        \langle\sv'| \VC\UC_{\theta}\WC\ket{\nv} = \langle \sv'_{A}|\otimes\langle \sv'_{B}| \VC\PC_{\theta}(\WC_A \otimes \WC_B)\ket{\nv_A}\otimes\ket{\nv_B} = \langle \sv'_A|\widetilde{\VC}_{\sv',\nv}\PC_\theta \WC_A\ket{\nv_A}
    \ee
    where $\widetilde{\VC}_{\sv',\nv} := \bra{\nv_B}\VC(\Ibb_A\otimes \WC_B)\ket{\nv_{i,B}} \in \Ubb(d_A)$ where $d_A$ is the dimension of the Hilbert space $\HC_A$ associated to region $A$, i.e. $d_A := \text{dim}(\HC_A) = \binom{n_A+m_A-1}{n_A}$. Then, by expanding the new amplitude in Eq.~\eqref{eq:reduced_amplitude}, the proofs are equivalent to Theorem~\ref{thm:GPSR} and Corollary~\ref{corollary:GPSR_odd} in the subspace $\HC_A$.
\end{proof}

\subsection{Number of shifts depends on the observable}\label{appendix:gpsr_reduced2}
In the following lemma, we highlight that for single-mode boson operators, number-ordered monomials can be expressed as linear combinations of normal-ordered ones. In doing so, we also show that the degree of the number-ordered operator corresponds to the highest degree appearing in the normal-ordered polynomial.
\begin{lemma}[Combinatorics of boson algebra, Ref.~\cite{katriel1974combinatorial}]\label{lemma:combinatorics_boson}
    A number-ordered monomial $O = \hat{n}^p$ is related to a polynomial in normal-ordered  operators of same degree $p$. In particular, the following identity holds:
    \be
        \hat{n}^p = \sum_{i=0}^n S(p,i) (\hat{a}\ad)^{i} \hat{a}^i\\,\qquad S(p,i) = \dfrac{1}{i!} \sum_{j=0}^i \binom{i}{j}(-1)^{i-j}j^p.
    \ee
    The coefficients $S(p,i)$ are Stirling numbers of the second kind.
\end{lemma}
\begin{proof}
    We consider the following
    \be
        \sum_{i=0}^{\infty} \dfrac{\lambda^i}{i!} (\hat{a}\ad)^{i} \hat{a}^i \ket{n} = \sum_{i=0}^\infty \lambda^i \binom{n}{i} \ket{n} = [1+\lambda]^n\ket{n} = [1+\lambda]^{\hat{n}}\ket{n}\\.
    \ee
    Let us now take $1+\lambda:=e^{\mu}$. Expand the two sides of the above equation in powers of $\mu$
    \begin{align}
        \sum_{i=0}^\infty \dfrac{(e^{\mu}-1)^i}{i!} (\hat{a}\ad)^{i} \hat{a}^i &= e^{\mu \hat{n}} \\
        \sum_{i=0}^\infty \sum_{j=0}^i \dfrac{1}{i!} \binom{i}{j}(-1)^{i-j} e^{\mu j} (\hat{a}\ad)^{i} \hat{a}^i &= \sum_{p=0}^\infty \mu^p \dfrac{\hat{n}^p}{p!} \\
        \sum_{p=0}^\infty\sum_{i=0}^\infty \sum_{j=0}^i \dfrac{1}{i!} \binom{i}{j}(-1)^{i-j} \mu^p \dfrac{j^p}{p!} (\hat{a}\ad)^{i} \hat{a}^i &= \sum_{p=0}^\infty \mu^p \dfrac{\hat{n}^p}{p!}\\.
    \end{align}
    At order $p$ in $\mu$ we find the equality
    \be
        \hat{n}^p = \sum_{i=0}^\infty S(p,i) (\hat{a}\ad)^{i} \hat{a}^i
    \ee
    Noting that $S(p,i) \neq 0$ for $1 \leq i \leq p$, we complete the proof.
\end{proof}
\begin{corollary}[Observable-dependent frequency spectrum]\label{thm:appendix_observable_frequency}
    Consider an initial state $\ket{\nv}$ with $n$ photons, a unitary mode transformation given by $U=VP_{\theta}W$ and an observable as defined in Eq.~\eqref{eq:observable}. Then, the number of frequencies will be given by
    \be
    R = \mathrm{min}\{p,n\}\,,
    \ee
    where we recall that $p$ is the overall degree of the observable in powers of average photon number observables, each mode contributing with $p_i$ and $\sum_i p_i = p$.
\end{corollary}

\begin{proof}
    We only need to focus on the case where $p<n$, since if $p \geq n$ we can resort to the result provided in Theorem \ref{thm:GPSR}. To start with, we need to investigate how each term $n_i^{p_i}$ transforms under $U$. We have that
    \begin{align}
        n_i^{p_i} = (a_i\ad a_i)^{p_i} \xmapsto{\hspace{4pt} U\hspace{4pt}} &((U^T\av\ad)_i(U\ad\av)_i)^{p_i} \\
        &= (\sum_{j,k}u_{ji}u^*_{ki}a_j\ad a_k)^{p_i} \\
        &= \sum_{\substack{\{\kappa_{jk}\},\\
        \sum_{j,k}\kappa_{jk} = p_i}}\binom{p_i}{\vec{\kappa}}\prod_{j,k}(u_{ji}u^*_{ki}a_j^\dagger a_k)^{\kappa_{jk}},
    \end{align}
    where $\binom{p_i}{\vec{\kappa}}=\frac{p_i!}{\kappa_{11}!\dots\kappa_{mm}!}$. \\
    In the first equality we used the definition of a  transformation under $U$, and in the second equality we used the generalized multinomial expansion. We can now put together all the terms contained in the observable defined in Eq.~\eqref{eq:observable}
    \begin{align}
        \mathfrak{f}(\nv,p) \xmapsto{\hspace{4pt} U\hspace{4pt}} \mathfrak{f}(\nv',p) &= \prod_{i} n_i'^{p_i} \\
        &= \prod_i \sum_{\substack{\{\kappa_{jk}\} \\
        \sum_{j,k}\kappa_{jk} = p_i}}\binom{p_i}{\vec{\kappa}}\prod_{j,k}(u_{ji}u^*_{ki}a_j\ad a_k)^{\kappa_{jk}} \\
        & = \sum_{\substack{\{\kappa_{jk}^{(i)}\} \\
        \sum_{j,k}\kappa_{jk}^{(i)} = p_i}} \prod_{i,j,k} \alpha_i(u_{ji}u^*_{ki}a_j\ad a_k)^{\kappa_{jk}^{(i)}}
    \end{align}
    where we defined $\alpha_i = \binom{p_i}{\vec{\kappa}^{(i)}}^{1/m^2}$. We remark that the summation is not anymore for a fixed $i$, but over all $i$. By noticing that, for each configuration $\{\kappa^{(i)}_{jk}\}$, we verify that $\sum_{i,j,k}\kappa^{(i)}_{jk} = \sum_ip_i = p$, it is clear that $p$ entries of the unitary transformation, 
    as well as $p$ entries of its complex conjugate, will be multiplied together. This implies that $\mathfrak{f}(\nv,p)\sim \{e^{\pm i\theta},\dots,e^{\pm ip\theta} \}$, hence the result.
\end{proof}
It may now become clear how to determine the number of frequencies for a given observable, even when the number of creation and annihilitation operators do not match. Consider in fact a hermitian observable such as the \textit{Bogolioubov} Hamiltonian
 \be
H := \sum_{i,j}h_{ij}a\ad_i a_j + g_{ij} a\ad_i a\ad_j + g_{ij}^* a_i a_j \,,
\ee
where $[h_{ij}]$ is a hermitian matrix and $[g_{ij}]$ is a symmetric matrix. In this case, a single term in the sum will not have the same number of creation and annihilation operators, e.g. $a_i\ad a_j\ad\,,\, a_i a_j$. We then extend the definition of the observable to a normal-ordered one
\be\label{eq:def_normal_observable}
\mathfrak{f}(\vec{a\ad},\vec{a},p) := f(a_1\ad)^{q_1} a^{r_1}_1 \dots (a_m\ad)^{q_{m}} a^{r_{m}}_m + f^*(a_1\ad)^{r_1} a^{q_1}_1 \dots (a_m\ad)^{r_{m}} a^{q_{m}}_m \,,f\in \Cbb \,,\qquad q=\sum_i q_i \,,\, r=\sum_i r_i
\ee
where $p= \max\{ q,r\}$.

\begin{corollary}
    Consider an initial state $\ket{\nv}$ with $n$ photons, a unitary mode transformation given by $U=VP_{\theta}W$ and an observable as defined in Eq.~\eqref{eq:def_normal_observable}. Then, the number of frequencies will be given by
    \be
    R = \mathrm{min}\{p,n\}\\,
    \ee
    where we recall that $p$ is the overall degree of the observable 
    \be
        p= \max\{q,r\} =\max\bigg\{\sum_i q_i,\sum_i r_i\bigg\}
    \ee
\end{corollary}
\begin{proof}
    Similarly to the prof of Theorem~\ref{thm:appendix_observable_frequency}, applying the mode transformation to the first term in Eq.~\eqref{eq:def_normal_observable} will result in a $q$-fold product of the matrix entries $u_{ij}$ and a $r$-fold product of the matrix entries $u_{ij}^*$, having therefore frequencies $\{e^{-ir\theta},\dots,e^{iq\theta}\}$. The second term will similarly have $\{e^{-iq\theta},\dots,e^{ir\theta}\}$. In total, the observable will have frequencies $\{e^{-ip\theta}, e^{ip\theta}\}$
\end{proof}
\end{document}